\documentclass[12pt]{article}
\usepackage[T1]{fontenc}
\usepackage[utf8]{inputenc}
\usepackage{caption}
\usepackage{amssymb,amsmath}
\usepackage{float}
\usepackage{geometry}
\usepackage{color}

\newcommand{\kolo}[1]{\!\stackrel{\circ}{#1}\!  \vphantom{#1}}

\newcommand{\mnabla}{\kolo{\nabla}}

\newcommand{\mGamma}{\kolo{\Gamma}}

\newcommand{\dd}{\text{d}}

\newcommand{\Lag}{\mathcal{L}}
\usepackage{amsthm}
\newtheorem{theorem}{Theorem}[section]
\newtheorem{lemma}[theorem]{Lemma}

\title{How the~non-metricity of the~connection\\ arises naturally in the~classical theory of gravity}

\author{Bart{\l}omiej B\k{a}k\footnote{Email: Bartlomiej.Bak@fuw.edu.pl}\\
Faculty of Physics, University of Warsaw,\\ Department of Mathematical Methods in Physics, Poland \\ and \\ Jerzy Kijowski\footnote{Email: kijowski@cft.edu.pl}\\
Center for Theoretical Physics\\
Polish Academy of Sciences, Warsaw, Poland}

\date{\today}

\begin{document}

\maketitle


\begin{abstract}
Spacetime geometry is described by two -- {\em a priori} independent -- geometric structures: the symmetric connection $\Gamma$ and the metric tensor $g$. Metricity condition of $\Gamma$ (i.e. $\nabla g = 0$) is implied by the Palatini variational principle, but only when the matter fields belong to an exceptional class. In case of a generic matter field, Palatini implies non-metricity of $\Gamma$. Traditionally, instead of the (1st order) Palatini principle, we use in this case the (2nd order) Hilbert principle, assuming metricity condition {\em a priori}. Unfortunately, the resulting right-hand side of the Einstein equations does not coincide with the matter energy-momentum tensor. We propose to treat seriously the Palatini-implied non-metric connection. The conventional Einstein's theory, rewritten in terms of this object, acquires a much simpler and universal structure. This approach opens a room for the description of the large scale effects in General Relativity (dark matter?, dark energy?), without resorting to purely phenomenological terms in the Lagrangian of gravitational field. All theories discussed in this paper belong to the standard General Relativity Theory, the only non-standard element being their (much simpler) mathematical formulation. As a mathematical bonus, we propose a new formalism in the calculus of variations, because in case of hyperbolic field theories the standard approach leads to nonsense conclusions.
\end{abstract}

\section{Introduction}\label{int}

Gravity is believed to describe the universal interaction between all possible kinds of matter. Mathematically, the above conjecture can be formulated as the following rule (often called the {\it minimal coupling} rule): for every matter field $\phi$, whose local dynamics in the flat Minkowski space, equipped with the flat Lorentzian metric $\eta_{\mu\nu}$, is described by the {\em matter} Lagrangian density:
\begin{equation}\label{Lmatt}
    {\cal L}_{matt} = {\cal L}_{matt} (\phi , \nabla \phi , \eta) \, ,
\end{equation}
its {\it global} dynamics in an arbitrary spacetime (equipped with a -- possibly non-flat -- metric $g_{\mu\nu}$), together with interaction between matter and geometry, is described by the following {\em total} (i.e. “matter + gravity”), metric Lagrangian density:
\begin{equation}\label{Ltot}
    {\cal L}_g := {\cal L}_H + {\cal L}_{matt} \, .
\end{equation}
Here, by
${\cal L}_H$ we denote the Hilbert Lagrangian density\footnote{In this paper we use a geometric system of physical units, where $c=1$ and $G=1$ are dimensionless numbers, see \cite{Gravitation}. To rewrite formula \eqref{LHilb} in an arbitrary system of units one has to multiply the dimensionless constant “$\pi \approx 3,14\dots$” appearing here, by the fundamental physical constant “$\frac G{c^4}$”.} (see \cite{hil}): \label{first}
\begin{equation}\label{LHilb}
    {\cal L}_H = {\cal L}_H (g, \partial g , \partial^2 g) :=
    \frac {\sqrt{|\det g|}}{16 \pi} \, \kolo R \, ,
\end{equation}
and $\stackrel{\circ}{R}$ denotes the scalar curvature of the metric $g$ or, more precisely, of the metric Levi-Civita connection
\begin{equation}\label{Gamma0}
    \stackrel{\circ \ }{\Gamma^{\kappa}}_{ \lambda\mu}:=
    \frac 12 \, g^{\kappa\sigma}  \left(
    g_{\sigma\lambda , \mu} + g_{\sigma\mu , \lambda} - g_{\lambda\mu , \sigma}
    \right) \, .
\end{equation}
The geometric structure of the Minkowski space in formula \eqref{Lmatt} (both the metric $\eta$ and its flat connection) has to be replaced by the actual, possibly non-flat, structure $(g,\stackrel{\circ \ }{\Gamma})$. Hence, ${\cal L}_{H}$ depends upon the metric components $g_{\mu\nu}$, together with their first and second derivatives, according to \eqref{LHilb}, whereas matter Lagrangian depends upon metric and its {\it first derivatives} only, contained in the Levi-Civita connection $\stackrel{\circ \ }{\Gamma}$:
\begin{equation}\label{ogo}
        {\cal L}_{matt} = {\cal L}_{matt} (\phi , \mnabla \phi , g)
        = {\cal L}_{matt} (\phi , \partial \phi , g, \partial g)\, .
\end{equation}

Physical intuitions concerning the fundamental conceptual structure of General Relativity Theory were built on specific examples: electrodynamics, scalar field and, especially, on mechanics of continuous media (see, e.g., \cite{kij-magli1997}, \cite{kij-magli1998}). In all these cases the matter Lagrangian \eqref{ogo} {\bf does not} depend upon connection (and, therefore, upon derivatives of the metric tensor), i.e. instead of \eqref{ogo}, we have:
\begin{equation}\label{Lmatt-specific}
        {\cal L}_{matt} = {\cal L}_{matt} (\phi , \partial \phi , g) \, .
\end{equation}
This, very special, property of the matter Lagrangian was {\em explicitly} assumed in most papers discussing the basic structures of the theory, like, e.g., the Palatini's outstanding paper \cite{palatini}. If the matter field belongs to this exceptional category, both matter and gravitational field equations (i.e.~Euler-Lagrange equations derived from the total Lagrangian density \eqref{Ltot}) can be rewritten as follows:
\begin{eqnarray}
  \partial_\lambda p^\lambda &=& \frac {\partial {\cal L}_{matt}}{\partial \phi} \, ,
  \label{E_L}\\
  {\rm where\ \ \ \ \  }
  p^\lambda &:=& \frac {\partial {\cal L}_{matt}}{\partial \phi_{ ,\lambda}}
  \ \ \ ; \ \ \ \phi_{ ,\lambda}:= \partial_\lambda \phi \, ,
  \label{mom_p}\\
  {\cal G}^{\mu\nu} &=& 8\pi \, {\cal T}^{\mu\nu} \, , \label{EE}\\
  {\rm where\ \ \ \ \  }{\cal T}^{\mu\nu} &:=&   2 \,  \frac {\partial {\cal L}_{matt}}{\partial g_{\mu\nu}} \, , \label{E_MOM}
\end{eqnarray}
with
\begin{equation}\label{calG}
    {\cal G}_{\mu\nu} = \sqrt{|\det g|} \, G_{\mu\nu} \,
\end{equation}
denoting the density of the Einstein tensor:
\begin{eqnarray}
     G_{\mu\nu} = R_{\mu\nu} - \frac 12\, g_{\mu\nu}\,  R \, .
\end{eqnarray}
Here, $R_{\mu\nu}$ denotes the Ricci tensor, whereas its trace $R = g^{\mu\nu} R_{\mu\nu}$ denotes the scalar curvature.

But what is most important in these exceptional examples, is the fact that the quantity \eqref{E_MOM}, defined in this way, can really be interpreted as the energy-momentum tensor of the matter field, because its “time-time-component” correctly describes the special-relativistic energy density of the matter field evolving over a non-dynamical (i.e. fixed {\it a priori}) geometric spacetime structure $(g,\stackrel{\circ \ }{\Gamma})$. This observation appears in the classical literature under the name of the so called Belinfante-Rosenfeld theorem (cf. \cite{Belinfante}, \cite{Belinfante2} and \cite{Rosenfeld}). Moreover, it is conserved (i.e. satisfies equation $\nabla_\mu {\cal T}^{\mu\nu} = 0$) due to solely matter field equations \eqref{E_L} -- \eqref{mom_p}, even if gravitational field equations \eqref{EE} are not satisfied. This property allowed Einstein to argue that the matter energy-momentum tensor ${\cal T}^{\mu\nu}$ acts as the “source of gravity”, just as the electric current $j^\lambda$ in electrodynamics can be considered  “source of electromagnetism” as soon as its conservation: $\partial_\lambda j^\lambda$ is fulfilled. This argument allowed him to decisively cross the line between “physics of flat space” and “physics of curved space”. But now, that the latter is already well established in the domain of relatively weak gravitational fields, the next frontier must be crossed: we would like to understand the domain of {\em very strong} fields, where the curvature of space-time far exceeds what we observe even in recent, daring space missions. In particular, we would like to understand the dark matter and dark energy phenomena that govern the large-scale structure of spacetime. For this purpose we must, probably, admit generic matter fields, not only those satisfying the very restrictive condition \eqref{Lmatt-specific} and, maybe, accept the generic non-metric connections. The latter issue was recently discussed thoroughly in paper~\cite{uni}.

Unfortunately, the simple scheme \eqref{E_L} -- \eqref{E_MOM} fails for a generic matter field (e.g. vectorial, spinorial or tensorial), when {\em covariant} (and not just {\em partial}) derivatives are necessary to define the {\em covariant} matter Lagrangian, i.e. when, instead of \eqref{Lmatt-specific}, we have:
\begin{eqnarray}
        {\cal L}_{matt} &=& {\cal L}_{matt} (\phi , \mnabla \phi , g) ={\cal L}_{matt} (\phi , \partial \phi , g,\,  \mGamma) = {\cal L}_{matt} (\phi , \partial \phi , g, \partial g) \, .\label{Lmatt-2wst}
\end{eqnarray}
In this case, variation of \eqref{Ltot} with respect to the metric $g_{\mu\nu}$ produces an extra term on the right-hand side of Einstein equation \eqref{E_MOM}.
Consequently, field equations \eqref{E_L}~--~\eqref{EE} of the theory remain unchanged, but the partial derivative \eqref{E_MOM} of the matter Lagrangian density, which provides the right-hand side of Einstein equation \eqref{EE}, must be replaced now by the so called “variational derivative”:
\begin{eqnarray}
  {\cal T}^{\mu\nu} &=&2\,\frac{\delta \Lag_{matt}}{\delta g_{\mu\nu}} =  2\,
  \left\{  \frac {\partial {\cal L}_{matt}}{\partial g_{\mu\nu}} -\partial_\kappa \frac {\partial {\cal L}_{matt}}{\partial g_{\mu\nu , \kappa}}
  \right\} \, , \label{senm}
\end{eqnarray}
where we use the following (“jet-oriented”) notation: $g_{\mu\nu , \kappa}:= \partial_\kappa g_{\mu\nu}$.
The second (extra) term contains not only the second order derivatives of the metric, but also second order derivatives of the matter field. This changes considerably the structure of Einstein equations, makes it extremely complicated and, above all, ruins the standard physical interpretation of the theory: “mass generates curvature”, which was proposed already by Riemann and Clifford (cf. \cite{clifford}). This is due to the fact that formula \eqref{senm} does not represent neither energy (mass) nor momentum of the matter field, because {\em Belinfante-Rosenfeld theorem} does not apply in this case and, consequently, \eqref{senm} differs from the energy-momentum tensor, defined properly as the Hamiltonian density generating evolution of the matter field in the Hamiltonian picture. Moreover, small perturbations of the metric (i.e. {\em gravitational waves}) propagate differently than the electromagnetic waves (i.e. they do not follow the light-cones of the metric). The simplest example of such a matter field, namely a vector field $\phi^\mu$, is discussed in Section \ref{Palatini}. Hence, the traditional, heuristic interpretation of Einstein equation (mass causes the spacetime curvature) fails in this case. To avoid this discrepancy, one should probably treat matter and gravity on a more equal footing, as already suggested strongly by Albert Einstein, who stressed many times in his papers (see, e.g. \cite{ein-autob}), that splitting physical reality into “matter” and “geometry” is, from the fundamental point of view, probably not justified.

A considerable simplification of the above conceptual puzzle is due to the use of the classical, well established “Palatini method of variation” (cf. \cite{palatini}). Here, variation of the total Lagrangian density \eqref{ogo} is performed with respect to the two, {\em a priori} independent, geometric structures: 1) the metric tensor $g$ and 2) the connection $\Gamma$:
\begin{equation}\label{afPA}
    {\cal L}_g := {\cal L}_H(g, \Gamma,\partial \Gamma) + {\cal L}_{matt}( \phi , \nabla \phi , g ,) \, .
\end{equation}
A specific sensitivity of the covariant derivatives $\nabla \phi$ towards the connection, written symbolically as
\begin{equation}\label{cov-r}
    \nabla  \phi = \partial  \phi  + `` \ \phi \ \cdot \ \Gamma \ "  \  ,
\end{equation}
is implied by geometric properties of the matter field $\phi$ (e.g.~vector, tensor or  other) and will be decisive for its gravitational properties.

The fundamental consequence of this approach is that, in a generic case of a Lagrangian density \eqref{afPA}, the resulting connection $\Gamma$ is (possibly) non-metric, i.e.~differs (possibly) from the Levi-Civita connection \eqref{Gamma0} in a way depending upon the above sensitivity of the matter field. In fact, the metricity condition for $\Gamma$ is obtained only for specific matter Lagrangians, namely those fulfilling \eqref{Lmatt-specific}. One can say shortly that, whereas matter causes the spacetime curvature, its specific sensitivity \eqref{cov-r} towards the connection causes the non-metricity of the connection. This {\bf is not} a new phenomenon, but is usually left unsaid by most authors.

\

In the present paper we show,  in a simple and natural way, how and why the non-metricity arises in a generic case of a matter Lagrangian \eqref{ogo}. For this purpose, we analyse the validity of the Palatini principle and prove that it leads, in a generic case, to the non-metric connection: the only exception are theories given by specific matter Lagrangian densities, namely those fulfilling \eqref{Lmatt-specific}, where the Palatini equation implies, indeed, the conventional metricity condition for the connection.

One might naively think that the non-metric connection theory would not be equivalent to Einstein's conventional theory, where the connection's metricity is assumed {\em a priori}. We show in this paper that such a naive conclusion is false. Indeed, when rewritten back, in terms of the metric connection, the non-metric theories discussed here assume their conventional, Einsteinian form, the only difference being how do we interpret various mathematical (computational) terms arising in field equations. We therefore emphasize that the theories discussed here do not belong to any of the “generalizations of General Relativity”, but to the standard Einstein theory of gravity!

The goal of this paper is to convince the reader that such a reformulation of a non-metric connection, back to the purely metric picture, is neither necessary nor useful. We show that the formulation of the conventional theory of gravity, allowing for (possibly) non-metric connection, implied by the consistent application of Palatini's “method of variations”, significantly simplifies the conceptual framework of this theory and its mathematical structure. Being geometrically natural and algebraically useful, the above non-metric connection deserves, maybe, more attention on the side of the physical interpretation of the theory as the “true geometric structure” of spacetime. Indeed, there are strong physical arguments (cf.~\cite{uni}) for such an extension of the geometric framework for the gravity theory, especially when trying to describe the large scale structure of the Universe. But in the present paper we consider only the conventional Einstein theory. Transition from its purely metric formulation to the metric-affine formulation (with a non-metric connection) is treated here as merely the mathematical “change of variables” which does not change the physical content of the theory.

A considerable simplification of the canonical
framework of the General Relativity Theory is due to the discovery of its purely affine variational formulation (see~\cite{affine}), where variation of the affine Lagrangian density
\begin{equation}\label{af}
    \Lag_A = \Lag_A(R_{\mu\nu} , \phi , \nabla \phi )
\end{equation}
is performed with respect to the connection variable only, whereas the metric tensor arises as the “momentum canonically conjugate” to the connection. Such a purely affine Lagrangian density depends upon connection coefficients $\Gamma^\lambda_{\ \mu\nu}$ (contained in the Ricci curvature tensor $R_{\mu\nu}$, but also in the matter covariant derivatives $\nabla \phi$), together with their first (and not the {\it second} !) derivatives contained only in the curvature tensor, (cf.~\cite{affine}, \cite{Tulcz}, \cite{kij-werp2007} and a recent review article \cite{APP}). We thoroughly analyse the relation between the affine formulation and the one due to Palatini.

\

The paper is organized as follows. In Section \ref{sec var} we show that the classical formalism in calculus of variations, based on “imposing spacetime-boundary-conditions”, contradicts the hyperbolic structure of field equation. This classical approach was developed for optimization purposes, described by elliptic equations. Applying it to dynamical (hyperbolic) problems is a pure nonsense. We show how to avoid this nonsense: the variational principle has to be treated as a “symplectic relation”. The approach presented here enables us to easily manipulate various variational principles and greatly simplifies the proof of their equivalence. In Section \ref{Palatini} we first analyse thoroughly the content of the simplified version of the “Palatini method of variations”, limiting ourselves to the case of “connection-non-sensitive” matter fields. This was the case considered by Palatini himself in his outstanding paper \cite{palatini}. Next, we show that -- in the case of a generic matter field -- it necessarily leads to a non-metric spacetime connection. We provide its complete formulation as a rigorously defined relation between the two geometric spacetime structures: the metric tensor and the affine connection. In Section \ref{univ-aff} we discuss the transition from the metric formulation to the affine formulation in terms of a simple Legendre transformation. In Section \ref{univ-P} we construct the universal Palatini formulation of the present General Relativity theory and in Section \ref{afin} we briefly discuss the inverse Legendre transformation: from affine to the metric picture. Finally, we discuss the physical consequences (Section \ref{concl}) of the mathematical results obtained here. In particular, we show which sector of the geometric structure of space-time could be responsible for the description of dark matter, if we try to generalize the current theory of gravity to the scale of the Universe. The most complex calculations have been shifted to Appendixes.

The reader is warned that, without loss of generality, the connection $\Gamma^\lambda_{\ \mu\nu}$ used in this paper is always symmetric. Indeed, a generic (non-symmetric) connection in the tangent bundle (a concept conceived for purposes of the Yang-Mills theory) is a {\em reducible} object. Due to the specific structure of the tangent bundle (the “solder form”) it splits into two irreducible parts: 1) a symmetric connection and 2) the torsion tensor. Being a tensor field, the latter can be always taken into account together as one of the matter fields. Therefore, without loss of generality, we limit ourselves to the case of a symmetric connection: $\Gamma^\lambda_{\ \mu\nu}= \Gamma^\lambda_{\ \nu\mu}$.

\section{Hyperbolic calculus of variations: common misconceptions and how to avoid them. Affine formulation of the theory of gravity}
\label{sec var}
Consider a Lagrangian density ${\cal L}= {\cal L}(\varphi^K , \partial_{\lambda}\varphi^K )$ depending upon a one-parameter-family of fields
\begin{eqnarray}
    \varphi^K = \varphi^K(x^\mu, \epsilon)\, ,
\end{eqnarray} 
where $(x^\mu)$ are spacetime coordinates. We denote by
\begin{eqnarray}
    \varphi^K_{,\lambda} := \partial_\lambda \varphi^K \, ,
\end{eqnarray} 
the field's spacetime derivatives.
Traditionally, derivative with respect to the parameter $\epsilon$ is denoted by $\delta$:
\begin{equation}\label{delta}
    \delta := \frac {\dd}{\dd  \epsilon}\, .
\end{equation}
Operator $\delta$ commutes, obviously, with spacetime derivatives:
\begin{eqnarray}
    \delta \left( \partial_\mu \varphi^K \right) =\delta \varphi^K_{,\mu} = \partial_\mu \delta \varphi^K \, ,
\end{eqnarray}
(this trivial observation is, in many textbooks, upgraded to the level of “the fundamental Lemma of the calculus of variations”). The following, obvious identity:
\begin{eqnarray}
  \delta \Lag &=& \frac {\partial \Lag}{\partial \varphi^K}\, \delta \varphi^K +
    \frac {\partial \Lag}{\partial \varphi^K_{,\lambda}}\, \delta \varphi^K_{,\lambda}= \left(\frac {\partial \Lag}{\partial \varphi^K} - \partial_\lambda
    \frac {\partial \Lag}{\partial \varphi^K_{,\lambda}}
    \right)\, \delta \varphi^K + \partial_\lambda \left(
    \frac {\partial \Lag}{\partial \varphi^K_{,\lambda}}\, \delta \varphi^K
    \right) \, ,  \qquad \label{E-L-deriv}
\end{eqnarray}
is the starting point of the calculus of variations. We call the first term of \eqref{E-L-deriv} the {\it volume part} (the \textit{bulk term}) and the second one the {\em boundary part}. Traditionally, one neglects the boundary part because, when integrated over a spacetime volume ${\cal O}$, imposing boundary conditions on its boundary $\partial {\cal O}$, implies $\delta \varphi^K |_{\partial {\cal O}} = 0$. This way, one derives the second-order partial differential equations for unknown functions $\varphi^K$, called {\em Euler-Lagrange equations}:
\begin{equation}\label{E-L}
    \frac {\partial \Lag}{\partial \varphi^K} - \partial_\lambda
    \frac {\partial \Lag}{\partial \varphi^K_{,\lambda}} = 0 \, ,
\end{equation}
as the necessary condition for the extremum of the functional:
\begin{equation}\label{intL}
    {\cal F} := \int_{\cal O} \Lag  \, ,
\end{equation}
within the functional-analytic space of fields fulfilling the fixed spacetime boundary conditions $\varphi^K |_{\partial {\cal O}} = f^K$.

For example, C. Misner, K. Thorne and J.A. Wheeler in their monograph \cite{Gravitation}, otherwise excellent as an introduction to General Relativity Theory, calculate only the volume part of the variation of the Hilbert Lagrangian density and neglect the boundary part. As a justification of such a careless procedure (see \cite{Gravitation}, page 520, just above their formula 21.86) they write:

\

\noindent
“Variation of the geometry interior to the boundary make no difference in the value of the surface term. Therefore, it has no influence on the equations of motion to drop the term (21.85)”.

\

The term, which is dropped there, is precisely the surface term. It was never calculated in this monograph (cf. also \cite{pieszy} for more details).

Such a procedure, whose origin goes back to Johann Bernoulli and his classical {\em brachistochrone problem} (1696), works perfectly for “optimization problems”, where Euler-Lagrange equations are elliptic. Unfortunately, it is entirely false in the case of dynamical theories, governed by hyperbolic -- not elliptic -- field equations. Most theoretical physicists have already learned long time ago that there is no extremum, but only a “saddle point”, in the hyperbolic case, and believe that the above observation represents the pinnacle of human understanding of the principles of variation.

Moreover, the “Feynman integral” quantization method is based precisely on the observation that the classical trajectory -- i.e. the action's stationary point -- gives the main contribution to the integral over classical trajectories. But the “extremum {\em versus} saddle point” dichotomy is not enough to describe the entire complexity of the variational problems, because  the real difficulty lies elsewhere! Namely: {\bf no matter whether we expect extremum or a saddle point}, imposing the spacetime-boundary conditions is {\it strictly forbidden} in hyperbolic theories. This means, that there is no solution of \eqref{E-L} for a generic choice of boundary data! To convince oneself that this is the case, it is enough to consider the “mother of all hyperbolic theories”, that is, the wave equation
\begin{equation}\label{wave}
    \left( \frac{\partial^2}{\partial x^2} - \frac{\partial^2}{\partial t^2} \right) \varphi = 0 \, ,
\end{equation}
in two-dimensional spacetime $\mathbb{R}^2=\{(t,x)\}$. Implying advanced and retarded coordinates $(u,v)=(t-x,t+x)$ and twice integrating it over the rectangle
\begin{eqnarray}
{\cal R}=\{(u,v)\in\mathbb{R}^2 : \, u_0\leq u \leq u_0+2\delta,\, v_0\leq v \leq v_0+2\epsilon\}\, ,
\end{eqnarray}
is easy to prove that field equation \eqref{wave} for the function $\varphi(u,v)$ is equivalent to the following identity
\begin{eqnarray}
   \varphi(u_0+2\delta,v_0+2\epsilon) - \varphi(u_0+2\delta,v_0) -\varphi(u_0,v_0+2\epsilon)+\varphi(u_0,v_0) = 0\, ,
\end{eqnarray}
for any choice of four numbers: $\{u_0,v_0,\epsilon,\delta \}$. The above equation  could be simply re-transformed to $(t,x)=(\frac{v+u}{2}, \frac{v-u}{2})$ coordinates. Then:
\begin{eqnarray}
    \varphi(t_0+\epsilon +\delta, x_0+\epsilon - \delta)- \varphi(t_0+\delta , x_0-\delta)  - \varphi(t_0+\epsilon, x_0+\epsilon) + \varphi(t_0,x_0)   &=& 0 \, . \qquad  \label{ident}
\end{eqnarray}
Putting $t_0=0$, $x_0=x$, $\delta = x$ and $\epsilon = 1-x$ we obtain an identity which must be fulfilled for any $0 \le x \le 1$:
\begin{equation}\label{iden-boundary}
  \varphi(1,1-x)- \varphi(x,0)  - \varphi(1-x,1) +  \varphi(0,x) = 0\, ,
\end{equation}

Consider now the spacetime volume ${\cal O} = [0,1] \times [0,1]$, i.e:
\begin{equation}\label{obszar}
        {\cal O} = \left\{ (t,x) \middle| 0 \le t \le 1 \ ; \ 0\le x \le 1 \right\} \, .
\end{equation}

We see (unfortunately, few physicists are aware of that!) that field equation implies here a constraint in space of boundary data: the value of the field on the upper wall (i.e.: $\varphi(1,\cdot)$) is uniquely given by its value on the remaining three walls (i.e.: $\varphi(\cdot , 0)$, $\varphi(\cdot , 1)$ and $\varphi(0, \cdot)$). There is no solution of the wave equation if the boundary data do not satisfy the constraint ${\cal C}$ defined by equation \eqref{iden-boundary}! Moreover: field equation \eqref{wave} is {\em equivalent} to this constraint!

Hence, the “brachistochrone” philosophy allows us to derive field equations {\bf provided we already know field equations}. By choosing the boundary data “at random”, the probability that there is any solution that satisfies this choice is zero.

We stress, that the constraint \eqref{iden-boundary} is still “relatively manageable” for the simple spacetime rectangle \eqref{obszar}, whereas for a generic spacetime volume ${\cal O}$ (e.g., a time slice $\{a \le t \le b ; x \in \mathbb{R}\}$) it is a much worse, very singular, non-close subspace in any reasonable topology of boundary data.

We conclude that, as a method of {\em deriving} field equation, the above “brachistochrone philosophy” breaks down completely in case of hyperbolic field equations. Furthermore, we stress, however, that  this method works perfectly for purposes of optimization problems governed by elliptic field equations, because boundary conditions can be imposed without any restriction in those cases.

Our conclusion does not mean that the formulae used in the Lagrangian field theory are false and useless! Below, we will give them a coherent mathematical meaning, which replaces the nonsensical heuristics based on the “spacetime boundary conditions” and “the least action principle”. The (obviously false) term “principle of least action” was popularized by Pierre Louis Maupertuis in an attempt to support the extremely optimistic philosophy of G. Leibniz, so much ridiculed by Voltaire in his novella \textit{Candide}.

For this purpose we propose to work “on shell”, i.e. to restrict oneself only to those field configurations (and their jets $(\varphi^K, \varphi^K_{,\lambda} )$) which fulfil field equations \eqref{E-L}. This means, that -- instead of neglecting the boundary part of \eqref{E-L-deriv} -- we neglect its volume part. This way, the formula \eqref{E-L-deriv} is no longer an identity, but becomes an equation imposed on the first jet of the field configuration:
\begin{eqnarray}
    \delta {\cal L}(\varphi^K , \varphi^K_{,\lambda} )  &=& \partial_\lambda \left(
    \frac {\partial {\cal L}}{\partial \varphi^K_{,\lambda}}\, \delta \varphi^K
    \right)  = \left( \partial_\lambda p_K^{\ \lambda} \right)\, \delta \varphi^K +
    p_K^{\ \lambda} \, \delta \varphi^K_{,\lambda}
    \, , \label{E-L-deriv1}
\end{eqnarray}
where the canonical momentum $p_K^{\ \lambda}$ has been introduced as a shortcut notation for the following expression:
\begin{equation}\label{mom}
    p_K^{\ \lambda} := \frac {\partial {\cal L}}{\partial \varphi^K_{,\lambda}} \, .
\end{equation}
We see that the system of first-order partial differential equations \eqref{E-L-deriv1} for the variables $\left(\varphi^K , p_K^{\ \lambda} \right)$ is equivalent to the second-order Euler-Lagrange equation \eqref{E-L}, written in the form:
\begin{equation}\label{E-L2}
    \partial_\lambda p_K^{\ \lambda} = \frac {\partial {\cal L}}{\partial \varphi^K }\, ,
\end{equation}
together with definition \eqref{mom} of momenta. This way, at each spacetime point ${\bf m}~=~(x^\mu)$, field equations \eqref{E-L-deriv1} can be considered as a symplectic relation (i.e. a Lagrangian submanifold) in a symplectic space ${\cal P}_{\bf m}$ parameterized by the following “generalized jets” of fields: $(\varphi^K ,\,  \varphi^K_{,\lambda} , \, p_K^{\ \lambda},\,  j_K:= \partial_\lambda p_K^{\ \lambda})$. Mathematically, this approach was rigorously defined in \cite{Tulcz}, \cite{CJK} and \cite{Kij-Moreno2015}, but its strength consists in the fact that it is very well adapted for practical calculations in both the Lagrangian and Hamiltonian formalism (especially when constraints are present) and avoids the nonsensical procedure of “imposing the spacetime-boundary conditions”. Practically, this approach is based on splitting canonical field variables into two groups: the “control parameters” (those, who appear under the sign “$\delta$” -- in case of \eqref{E-L-deriv1} these are configuration variables $\varphi^K$ and their “velocities” $\varphi^K_{,\lambda}$) and the “response parameters” (in case of \eqref{E-L-deriv1} these are momenta $p_K^{\ \lambda}$ and their “currents” $j^K= \partial_\lambda p_K^{\ \lambda}$). Field equations are then treated as the “control -- response relation”. We shall use this simplified formalism in the sequel.

These techniques were informally present already in classical texts, written by Lagrange, Carath{\'e}odory and other pioneers of the calculus of variations, and also in classical thermodynamics. As an example, consider the classical, thermodynamical formula:
\begin{equation}
    \delta U(V, S) = -p\, \delta V + T\, \delta S\, ,
\end{equation}
equivalent to:
\begin{equation}
    p = - \frac {\partial U}{\partial V} \, , \qquad T = \frac {\partial U}{\partial S}
    \, ,
\end{equation}
which selects the two-dimensional subspace of all the physically admissible states of a simple thermodynamical body, as a Lagrangian submanifold within a four-dimensional symplectic manifold parameterized by $(V,S,p,T)$ (volume, entropy, pressure, temperature) and equipped with the canonical (“God given”) symplectic form:
\begin{equation}\label{Omega}
    \omega = - \delta p \wedge \delta V + \delta T \wedge \delta S \, .
\end{equation}
Similarly, classical mechanics can be formulated as a symplectic relation:
\begin{equation}
    \delta L(q, \dot{q}) = \frac {{\rm d}}{{\rm d}t} \left( p\,  \delta q\right)
     = \dot{p} \, \delta q + p\,  \delta \dot{q}\, ,
\end{equation}
(equivalent to: $\dot{p} =  \frac {\partial L}{\partial q}$, $p = \frac {\partial L}{\partial \dot{q}}$) with respect to the canonical symplectic form:
\begin{equation}\label{Omega1}
    \omega = \frac {{\rm d}}{{\rm d}t} \left( \delta p \wedge \delta q\right) =
    \delta \dot{p} \wedge \delta q + \delta p \wedge \delta \dot{q} \, .
\end{equation}
Legendre transformations, like transition from adiabatic to the thermostatic insulation in thermodynamics, or from the Lagrangian to the Hamiltonian picture in mechanics, are simply described in this formalism as an exchange between control and response parameters:  $T$ {\it versus} $S$ in \eqref{Omega} and $p$ {\it versus} $\dot{q}$ in \eqref{Omega1}.

In case of the gravitational field, the missing boundary term in the Wheeler-Misner-Thorn formula (21.86) was calculated in paper \cite{pieszy} and will be presented below. Following V.A.~Fock, who realized that the substantial simplification of the canonical structure of gravity theory is obtained when, instead of the covariant tensor $g_{\mu\nu}$, one represents the metric structure of spacetime by its contravariant density (see \cite{Fock}),  we introduce the following notation\footnote{See footnote on page (\ref{first}) for physical units used here.}:
\begin{equation}\label{pi2}
{\pi}^{\mu\nu} := \frac 1{16 \pi} \sqrt{|\det g|} \  g^{\mu\nu}
    = \frac {\partial {\cal L}_H}{\partial R_{\mu\nu}} \, .
\end{equation}
With respect to the Fock's book \cite{Fock}, our modest contribution here consists in incorporating the gravitational constant -- i.e. $\frac 1{16 \pi}$ in geometric units -- into the “momentum” variable $\pi$. Moreover, the following object arises automatically in the formula for the variation:
\begin{equation}\label{pi4}
{\pi}_{\lambda}^{\ \mu\nu\kappa} =
\frac {\partial {\cal L}_H}{\partial \Gamma^\lambda_{\ \mu\nu,\kappa}}
 =
{\pi}^{\mu\nu} \, \delta^\kappa_\lambda -\delta_{\lambda}^{(\mu}\, \pi^{\nu)\kappa}  \, .
\end{equation}
The Hilbert Lagrangian density assumes now the following form:
\begin{eqnarray}
  {\cal L}_H  =  {\cal L}_H (g,   \Gamma , \partial \Gamma) :=
    \frac {\sqrt{|\det g|}}{16 \pi} R =  \pi^{\mu\nu} R^\lambda_{\ \mu\lambda\nu} =\pi^{\mu\nu} R_{\mu\nu} \, ,
    \label{Lag H}
\end{eqnarray}
    where the Riemann and Ricci tensors are defined as usual:
\begin{eqnarray}
  R^\lambda_{\ \mu\nu\kappa} &:=&  \Gamma^\lambda_{\ \mu\kappa,\nu} -
     \Gamma^\lambda_{\ \mu\nu,\kappa} +
   \Gamma^\lambda_{\ \alpha\nu}\,\Gamma^\alpha_{\ \mu\kappa} -
   \Gamma^\lambda_{\ \alpha\kappa}\,  \Gamma^\alpha_{\ \mu\nu}\, ,
   \label{def riemann} \qquad  \\
      R_{\mu\nu} &:=&  R^\lambda_{\ \mu\lambda\nu} \label{def ricci}\, .
\end{eqnarray}
The missing boundary term in the variation of ${\cal L}_H$ follows from an identity, which is universally valid (see  \cite{pieszy} for the proof) for an arbitrary metric tensor $g$ and an arbitrary symmetric connection $\Gamma$ (not necessarily metric):
\begin{eqnarray}
\delta {\cal L}_H &=& -\
\frac 1{16 \pi}\, {\cal G}^{\mu\nu}\, \delta g_{\mu\nu} -
\left( \nabla_\kappa {\pi}_{\lambda}^{\ \mu\nu\kappa} \right) \, \delta
{\Gamma}^{\lambda}_{\ \mu\nu}+\partial_\kappa \left( {\pi}_{\lambda}^{\ \mu\nu\kappa} \ \delta
{\Gamma}^{\lambda}_{\mu\nu} \right) \, .
\label{variation_full}
\end{eqnarray}
We see, that variation of the Hilbert Lagrangian \eqref{LHilb} with respect to the connection can be easily calculated:
\begin{equation}\label{deltaLHil}
    \frac {\delta {\cal L}_H}{\delta {\Gamma^{\lambda}}_{ \mu\nu}} =-
    \nabla_\kappa {\pi}_{\lambda}^{\ \mu\nu\kappa} \, ,
\end{equation}
where $\nabla$ denotes the covariant derivative with respect to $\Gamma$ (see again \cite{pieszy}). An obvious algebraic identity:
\begin{equation}\label{d_detg}
    \delta \left( \sqrt{|\det g|} \right) =
    \frac 12 \sqrt{|\det g|} \  g^{\alpha\beta} \ \delta g_{\alpha\beta} \, ,
\end{equation}
implies yet another representation of variation of the Hilbert Lagrangian density:
\begin{eqnarray}
  R_{\mu\nu}\,  \delta \pi^{\mu\nu} &=& \frac 1{16 \pi} \, R_{\mu\nu}\,  \delta
  \left( \sqrt{|\det g| } \, g^{\mu\nu} \right)= \nonumber \\
    &=&  \frac {\sqrt{|\det g|}}{16 \pi}\, \left( R_{\mu\nu} \, \delta g^{\mu\nu} +
    \frac12 \,  R\,  g^{\alpha\beta}\,  \delta g_{\alpha\beta}\right)= \nonumber \\
    &=& - \frac {\sqrt{|\det g|}}{16 \pi}\, \left(
    R^{\mu\nu} -  \frac 12\, R \,g^{\mu\nu} \right)\, \delta g_{\mu\nu}
    =  -  \frac {1}{16 \pi}\,
    {\cal G}^{\mu\nu}\, \delta g_{\mu\nu}    \, , \label{R delta pi}
\end{eqnarray}
which enables us to rewrite \eqref{variation_full} in an equivalent form:
\begin{eqnarray}
\delta {\cal L}_H &=& R_{\mu\nu} \, \delta \pi^{\mu\nu} -
\left( \nabla_\kappa {\pi}_{\lambda}^{\ \mu\nu\kappa} \right)\,  \delta
{\Gamma}^{\lambda}_{\ \mu\nu} +  \partial_\kappa \left( {\pi}_{\lambda}^{\ \mu\nu\kappa}\,  \delta
{\Gamma}^{\lambda}_{\ \mu\nu} \right) \, .
\label{variation_full_pi}
\end{eqnarray}
In case of the purely metric Hilbert Lagrangian \eqref{LHilb} without any matter, the second term in both \eqref{variation_full} and \eqref{variation_full_pi} vanishes automatically -- see definitions \eqref{pi2}, \eqref{pi4} -- which implies the following field equation:
\begin{eqnarray}
     {\nabla}_\kappa {\pi}_{\lambda}^{\ \mu\nu\kappa} = 0 \  \Longleftrightarrow\  \nabla_{\kappa} g_{\mu\nu} = 0 \  \Longleftrightarrow\ \Gamma^{\lambda}_{\ \mu\nu}=\ \mGamma^{\lambda}_{\ \mu\nu}\, . \label{metricity0}
\end{eqnarray}
Hence, we end up with:
\begin{eqnarray}
  \delta {\cal L}_H (g, \, \partial g, \, \partial^2 g) &=&  - \frac 1{16 \pi}\, \kolo{\cal G}^{\mu\nu}\,  \delta g_{\mu\nu} + \partial_\kappa \left( {\pi}_{\lambda}^{\ \mu\nu\kappa} \, \delta \mGamma^{\kappa}_{ \ \lambda\mu} \right)=
\label{varG}\\
 &=& \kolo{R}_{\mu\nu}\, \delta \pi^{\mu\nu} + \partial_\kappa \left( {\pi}_{\lambda}^{\ \mu\nu\kappa}\, \delta \mGamma^{\kappa}_{ \ \lambda\mu} \right) \, ,
\label{varGr}
\end{eqnarray}
where the last, boundary term in both (equivalent) formulae: \eqref{varG} and \eqref{varGr}, represents the missing term in the Wheeler-Misner-Thorn formula (21.86), page 520 (a circle above geometric objects denotes their metricity, but the formula is valid for an arbitrary symmetric connection $\Gamma^{\kappa}_{ \ \lambda\mu}$ too).

The boundary term $\partial_\kappa \left( {\pi}_{\lambda}^{\ \mu\nu\kappa} \delta \Gamma^{\lambda}_{\ \mu\nu} \right)$ in the variational formula provides a strong argument for the affine approach, where connection $\Gamma^{\kappa}_{ \ \lambda\mu}$ plays role of the gravitational field configuration, whereas the metric tensor, encoded by the tensor-density~$\pi$, plays role of its canonically conjugate momentum, according to formulae \eqref{pi4} and \eqref{pi2}.

There is also a strong {\em physical} argument, based on the Newton's First Law, for choosing the connection $\Gamma$ (instead of the metric tensor) as the fundamental configuration variable of the gravitational field (see \cite{Senger1}).

\section{Mathematical structure of the Palatini variational principle. Emergence of non-metricity}\label{Palatini}

Below, we are going to use the above formalism of “symplectic relations” (in contrast to  “the least action principle”, based on the “spacetime boundary value problem”, which is strictly forbidden by the mathematical structure of the theory).

The modern Palatini approach is based on the following observation: derivatives of the metric enter linearly into the Levi-Civita connection \eqref{Gamma0} and, whence, to calculate variation of \eqref{Ltot} it is useful to “change variables” in space of second jets of the metric: from $(g,\, \partial g,\, \partial^2 g)$ to $(g, \, \Gamma,\, \partial\Gamma)$. Next, one can treat both the metric tensor and the connection coefficients as independent quantities: we do not assume  {\it a priori} the metricity condition \eqref{metricity0} of the connection, but derive it as one of the Euler-Lagrange equations. This simple trick is often called “the Palatini method of variation”, although it has been earlier used by Hilbert, Weyl, and Einstein himself. In fact, the originality of the Palatini paper \cite{palatini} with respect to these authors consists in the fact that, when calculating variation $\delta \Lag_H$, he was able to select properly the contribution due to the variation $\delta \Gamma$ of connection. However, in those days, the connection was not regarded as an independent geometric object: only “Christoffel symbols” were known, equivalent to our “metric connection” \ $\mGamma$. This fact obscures considerably the understanding of the Palatini's contribution.

We see that the contravariant density of metric $\pi^{\mu\nu}$  \eqref{pi2} plays the role of the “momentum canonically conjugate to the connection” (derivative of the Lagrangian density with respect to the derivatives of the connection). As already mentioned in the previous Section (see formula \eqref{deltaLHil}), variation of the Hilbert Lagrangian  with respect to the connection is following:
\begin{equation}
    \frac {\delta {\cal L}_H}{\delta {\Gamma^{\lambda}}_{ \mu\nu}} =
   - \nabla_\kappa {\pi}_{\lambda}^{\ \mu\nu\kappa} \, ,
\end{equation}
where $\nabla$ denotes the covariant derivative with respect to the symmetric connection~$\Gamma$. Consequently, variation of the total metric Lagrangian density ${\cal L}_g = {\cal L}_{matt} + {\cal L}_H$ equals:
\begin{eqnarray}
  0 = \frac {\delta \Lag_g}{\delta \Gamma^{\lambda}_{\ \mu\nu}} &:=&
  \frac {\partial \Lag_g}{\partial \Gamma^{\lambda}_{\ \mu\nu}} -
    \partial_\kappa \frac {\partial \Lag_g}{\partial \Gamma^{\lambda}_{\ \mu\nu\kappa}} = \frac {\partial \Lag_{matt}}{\partial \Gamma^{\lambda}_{\ \mu\nu}}
    - \nabla_\kappa {\pi}_{\lambda}^{\ \mu\nu\kappa}\, ,
    \label{nm} \\
  0 = \frac {\delta \Lag_g}{\delta g_{\mu\nu}} &:=&
  \frac {\partial \Lag_g}{\partial g_{\mu\nu}} -
    \partial_\kappa \frac {\partial \Lag_g}{\partial g_{\mu\nu\kappa}}=
    \frac {\partial \Lag_g}{\partial g_{\mu\nu}}=  \frac 1{16 \pi}
    \, \left( 8 \pi {\cal T}^{\mu\nu} - {\cal G}^{\mu\nu} \right) \, , \label{ei}
\end{eqnarray}
equivalently:
\begin{eqnarray}
 \frac {\partial \Lag_{matt}}{\partial \Gamma^{\lambda}_{\ \mu\nu}}
    &=&\nabla_\kappa {\pi}_{\lambda}^{\ \mu\nu\kappa}\, , \label{eqq1} \\
     \mathcal{T}^{\mu\nu}=: 2\,  \frac {\partial \Lag_{matt}}{\partial g_{\mu\nu}}&=&  \frac{1}{8\pi}\, {\cal G}^{\mu\nu}\, . \label{eqq2}
\end{eqnarray}
We see that, in particular case of a connection-independent-matter-Lagrangian \eqref{Lmatt-specific}, we have
\begin{equation}\label{noGamma}
    \frac {\partial \Lag_{matt}}{\partial \Gamma^{\lambda}_{\ \mu\nu}} = 0 \, ,
\end{equation}
and, whence, Euler-Lagrange equation \eqref{nm} (variation with respect to the connection) reduces (see eq.~\eqref{metricity0}) to metricity condition $\Gamma = \ \mGamma$.
Hence, what was assumed {\it a priori} in the purely metric approach, here, in the Palatini approach (i.e. in the mixed -- metric-affine approach), is obtained as one of field equations, i.e. as a result of the variational principle.

\

{\bf Example:} In electrodynamics, which was historically the first example, carefully analysed by Hilbert (see \cite{hil}), we have:
\begin{eqnarray}
    {\cal L}_{matt}     = - \frac 14 \sqrt{|\det g|}\, f_{\mu\nu}\,  f_{\alpha\beta}\,  g^{\mu\alpha}\,  g^{\nu\beta}\, ,
\end{eqnarray}
where
\begin{eqnarray}
    f_{\mu\nu} = \partial_\mu A_\nu - \partial_\nu A_\mu = 2 A_{[\nu , \mu ]}\, ,
\end{eqnarray}
denotes the Faraday 2-form. Thus, the electromagnetic four-potential $A_\mu$ plays the role of the matter field $\phi$. We see that this matter Lagrangian fulfils \eqref{noGamma} and, whence, metricity condition can be either assumed {\em a priori} (metric picture) or obtained as one of the field equations (Palatini picture).

 According to equation~\eqref{mom_p}, the role of the momentum “$p^{\mu\lambda}$”, canonically conjugate to the “matter variable” $A_\mu$, is assumed by the  contravariant tensor-density:
 \begin{eqnarray}
     {\cal F}^{\mu\lambda} := \frac {\partial {\cal L}_{matt}}{\partial A_{\mu ,\lambda}}
    = \sqrt{|\det g|}\, f_{\alpha\beta} \, g^{\alpha\mu}\, g^{\beta\lambda} =
    \sqrt{|\det g|}\, f^{\mu\lambda} \, .
 \end{eqnarray}
Hence, Euler-Lagrange equation \eqref{E_L} encodes the Maxwell equations $\partial_{\nu} {\cal F}^{\mu\nu} = 0$, whereas~\eqref{E_MOM} becomes the Maxwell (symmetric!) energy-momentum tensor density
\begin{equation}\label{E-M-Maxw}
        {\cal T}^{\mu\nu} = 2 \,\frac {\partial {\cal L}_{matt}}{\partial g_{\mu\nu}} =
    \sqrt{|\det g|}\, \left[f^{\mu}_{\ \beta}\,  f^{\nu \beta} - \frac 14 \, g^{\mu\nu}\, f_{\alpha\beta} \, f^{\alpha\beta}
    \right]\, .
\end{equation}
This means that
\begin{eqnarray}
{\cal T}^{\mu\nu} =
  \sqrt{|\det g|} \  T^{\mu\nu} \, ,    
\end{eqnarray}
with:
\begin{eqnarray}
    T^{\mu\nu} = \frac 2{\sqrt{|\det g|}} \,\frac {\partial {\cal L}_{matt}}{\partial g_{\mu\nu}} =
    f^{\mu}_{\ \beta}\,  f^{\nu \beta} - \frac 14 \, g^{\mu\nu}\, f_{\alpha\beta} \, f^{\alpha\beta}
    \, .
\end{eqnarray}
There is no doubt that, indeed, this quantity describes the energy-momentum density carried by the Maxwell field (cf. \cite{Gravitation}, \cite{pieszy}).

\

Unfortunately, the naive implementation of the above “Palatini method” fails when the matter Lagrangian depends upon connection coefficients $\Gamma$, which are contained in covariant derivatives of the matter fields.
This happens for generic matter fields, like vector, spinor or tensor fields, where covariant derivatives are necessary as the “building blocks” of the coordinate-invariant matter Lagrangian density. Hence, in a generic case, we have:
\begin{equation}\label{Lmatt-2}
        {\cal L}_{matt} = {\cal L}_{matt} (\phi , \partial \phi , g, \Gamma) \, ,
\end{equation}
and, consequently:
\begin{equation}\label{calP_10}
    {\cal P}_{\ \ \lambda}^{\mu\nu} := \frac {\partial {\cal L}_{matt}}
    {\partial  \Gamma^{\lambda}_{\ \mu\nu} } \ne 0 \, .
\end{equation}
We see, that the metricity of the connection (equation \eqref{metricity0}) {\bf is not} recovered! Instead, we would have obtained the following value of the non-metricity of the connection:
\begin{equation}\label{nonmetr}
    \nabla_\kappa {\pi}_{\lambda}^{\ \mu\nu\kappa} = {\cal P}_{\ \ \lambda}^{\mu\nu}
    :=  \frac {\partial {\cal L}_{matt}}
    {\partial \Gamma^{\lambda}_{\ \mu\nu} } \ne 0 \, .
\end{equation}
Given the value of the (newly defined above) field ${\cal P}_{\ \ \lambda}^{\mu\nu}$, this equation can be easily solved with respect to the connection $\Gamma$, i.e. $\Gamma$ can be uniquely reconstructed in the form:
\begin{equation}\label{Gamma-gen}
    \Gamma^\lambda_{\ \mu\nu} = \ \mGamma^{\lambda}_{\
    \mu\nu} +
    N^{\lambda}_{ \ \mu\nu} \, .
\end{equation}
This way, the covariant derivative of the metric: $\nabla \pi$, calculated with respect to the connection $\Gamma$, splits into the sum: the covariant derivative $\ \mnabla \pi$, which vanishes identically, plus a combination of $N$'s multiplied by $\pi$. Finally, we obtain the following, linear equation for $N^{\lambda}_{ \ \mu\nu}$:
\begin{eqnarray}
    {\cal P}_{\ \ \lambda}^{\mu\nu} &=& \pi^{\mu\alpha} N^{\nu}_{ \ \lambda\alpha} +
    \pi^{\nu\alpha} N^{\mu}_{ \ \lambda\alpha} - \pi^{\mu\nu} N^{\alpha}_{ \ \lambda\alpha} - \frac 12
    \left( \delta^\mu_\lambda N^{\nu}_{ \ \alpha\beta} + \delta^\nu_\lambda N^{\mu}_{ \ \alpha\beta}
    \right) \pi^{\alpha\beta} \, ,\label{N=lincalP}
\end{eqnarray}
which can easily be solved:
\begin{eqnarray}
  N_{\kappa\lambda\mu} &=&
 \frac{8\pi}{\sqrt{|\det g|}}\, \left[\mathcal{P}_{\kappa\lambda\mu} + \mathcal{P}_{\kappa\mu\lambda} - \mathcal{P}_{\lambda\mu\kappa} + g_{\mu\lambda}\, \left(\frac 12\, \mathcal{P}^{\sigma}_{\ \sigma \kappa} -\mathcal{P}_{\kappa\sigma}^{\ \  \sigma} \right)+\right. \nonumber \\
    &&\qquad \qquad+\left.\frac{2}{3}\, g_{\kappa(\lambda}\, \mathcal{P}_{\mu)\sigma}^{\ \ \ \sigma}- \mathcal{P}^{\sigma}_{\ \sigma(\lambda}\, g_{\mu)\kappa} \right] \, .
 \label{nonmetricity}
\end{eqnarray}
(proof in the Appendix \ref{calP N}, see also Appendix \ref{proof lemma non-metricity}).

\

In the present paper we show that it is worthwhile to use the above non-metric connection \eqref{Gamma-gen}, arising from the naive implementation \eqref{nonmetr} of the “Palatini method of variation”, because it simplifies considerably the standard description of the canonical structure of General Relativity Theory. In particular, it is precisely the one, which describes properly the field energy (cf.~\cite{Belinfante}, \cite{Belinfante2}, \cite{Rosenfeld}) and also arises in its “purely affine” formulation (cf.~\cite{affine}, \cite{pieszy}, \cite{kij-werp2007}).

\subsection{Simplified version of the Palatini principle}

To analyse better the mathematical structure of the Palatini principle, let us begin with the special case \eqref{Lmatt-specific}, where the matter Lagrangian density ${\cal L}_{matt}$ does not depend upon derivatives of the metric tensor (i.e.~upon connection $\mGamma$). We see that the “on shell” variation of the matter Lagrangian density can be written as:
\begin{eqnarray}
  \delta {\cal L}_{matt} (\phi , \partial \phi , g)  =  \frac {\partial {\cal L}_{matt}}{\partial g_{\mu\nu}}\, \delta g_{\mu\nu} +
    \partial_\lambda  \left( p^\lambda\,  \delta \phi \right) \, . \label{gen_1}
\end{eqnarray}

Because ${\cal L}_{matt}$ does not depend upon derivatives of the metric tensor (i.e.~upon connection), we can add an extra, trivial term to formula \eqref{gen_1}:
\begin{eqnarray}
    \delta {\cal L}_{matt} (\phi , \partial \phi , g) &=&
    \frac {\partial {\cal L}_{matt}}{\partial g_{\mu\nu}}\,  \delta g_{\mu\nu} + {\cal P}_{\ \ \lambda}^{ \mu\nu} \, \delta
    {\Gamma^{\lambda}}_{ \mu\nu} +     \partial_\lambda \left( p^\lambda\, \delta \phi \right) \, ,\label{gen_2}
\end{eqnarray}
because it is equivalent to an extra, trivial field equation:
\begin{equation}\label{calP_1}
    {\cal P}_{\ \ \lambda}^{\mu\nu} := \frac {\partial {\cal L}_{matt}}
    {\partial {\Gamma^{\lambda}}_{ \mu\nu} } = 0 \, .
\end{equation}
Formulae \eqref{gen_2} and \eqref{calP_1} provide the simplest proof of the validity of the “Palatini method of variation” for a Lagrangian density of the type \eqref{Lmatt-specific}: instead of the second order variation of \eqref{Ltot} with respect to the metric $g$, we may equivalently perform the first order variation with respect to the two ({\em a priori} independent) geometric fields: $g$ and $\Gamma$. Indeed, defining the Palatini Lagrangian density:
\begin{eqnarray}
    &&\Lag_P(\phi , \partial \phi , g, \Gamma , \partial \Gamma) :=  \Lag_{matt}(\phi , \partial \phi , g) + \frac {\sqrt{\det g}}{16 \pi} \, R(g, \Gamma , \partial \Gamma) \, ,\label{Pal}
\end{eqnarray}
and using \eqref{variation_full} together with \eqref{gen_2}, we obtain:
\begin{eqnarray}
  \delta \Lag_P &=& \frac 1{16 \pi} \left( 8 \pi {\cal T}^{\mu\nu} -
  {\cal G}^{\mu\nu} \right)\, \delta g_{\mu\nu} +
  \left({\cal P}_{\ \ \lambda}^{ \mu\nu} - \nabla_\kappa {\pi}_{\lambda}^{\ \mu\nu\kappa} \right)\, \delta
{\Gamma}^{\lambda}_{\ \mu\nu}+ \nonumber \\
   & &  + \partial_\kappa \left( p^\kappa\,  \delta \phi +
   {\pi}_{\lambda}^{\ \mu\nu\kappa}\,  \delta
{\Gamma}^{\lambda}_{\ \mu\nu} \right)\, , \label{var LP}
\end{eqnarray}
which simply means that its “variational derivatives” are equal to \eqref{nm} -- \eqref{ei}. Moreover, because quantity ${\cal P}$ vanishes (see \eqref{calP_1}), equation \eqref{nm} reduces to \eqref{metricity0}. This means that the metricity condition of the connection -- which was not assumed {\it a priori} -- is obtained as one of the Euler-Lagrange equations of theory, namely equation \eqref{metricity0}.

\

We conclude, that in case of a special matter Lagrangian ${\cal L}_{matt}$, fulfilling~\eqref{Lmatt-specific} (like electrodynamics), the original, purely metric, second order variational principle \eqref{Ltot}, where variation is performed with respect to the metric tensor exclusively, is equivalent to the first order “Palatini variational principle”, where variation is performed with respect to both geometric quantities: connection and metric, independently.

\subsection{Generic case}

The above equivalence of the two “methods of variation” {\bf does not hold} {\em a priori} in a generic case of a Lagrangian density ${\cal L}_{matt} (\phi , \partial \phi , g, \mGamma)$, i.e.~where \eqref{calP_1} is no longer true, because the  corresponding Euler-Lagrange equation, resulting from variation with respect to $\Gamma$:
\begin{equation}\label{non-m}
    \nabla_\kappa {\pi}_{\lambda}^{\ \mu\nu\kappa}
    =
   {\cal P}_{\ \ \lambda}^{\mu\nu} := \frac {\partial {\cal L}_{matt}}
    {\partial \mGamma^{\lambda}_{\ \mu\nu} } \ne 0
    \, ,
\end{equation}
would imply the non-metricity of the connection, i.e.~the theory which is -- naively -- non-equivalent to the original Einstein theory.

\

{\bf Example:} If $\phi= (\phi^\alpha )$ is a vector field, then the unique way to construct an invariant scalar out of derivatives of $\phi$ is to use covariant derivatives:
\[
    \mnabla_{\beta} \phi^\alpha = \partial_{\beta} \phi^\alpha +
    \mGamma^{\alpha}_{\ \beta\sigma}\, \phi^\sigma \, ,
\]
and, whence, according to \eqref{non-m} we have the non-metricity tensor, which does not vanish identically:
\begin{eqnarray}
    {\cal P}_{\ \ \kappa}^{ \lambda\mu} &=& \frac {\partial {\cal L}_{matt}}{\partial \mGamma^{\kappa}_{\ \lambda\mu}} =\frac {\partial {\cal L}_{matt}}{\partial \left(\mnabla_{\beta} \phi^\alpha\right)} \, \frac{\partial \left(\mnabla_{\beta} \phi^\alpha \right)}{\partial \mGamma^{\kappa}_{\ \lambda\mu}}   = \frac {\partial {\cal L}_{matt}}{\partial \phi^\alpha_{\ ,\beta}}\, \delta^{\alpha}_{\kappa}\, \delta^{(\lambda}_{\beta}\, \delta^{\mu)}_{\sigma}\, \phi^{\sigma}   =
    p_\kappa^{\ (\lambda}\, \phi^{\mu)}   \, .\label{calP_3n}
\end{eqnarray}
We see, that the naive (i.e.~straightforward) implementation of the “Palatini method of variation” implies the non-metricity of the connection and, consequently, field equations which are {\em a priori} non-equivalent with the original metric theory.

However, we are going to prove in this paper, that the universal “Palatini variational principle”, consisting in varying with respect to independent geometric fields $\Gamma$ and $g$, {\it does exist}, if we only accept the above non-metricity, defined by equation~\eqref{non-m}. The resulting field theory {\bf is not} a new physical theory, but merely a mathematically equivalent reformulation of the standard {\it purely metric} theory, based on Lagrangian density \eqref{Ltot}. Even if formulated in terms of non-metric connection, it is completely equivalent to the original metric theory when recalculated in terms of the metric $g$ and its derivatives. This equivalence is the main result of this paper and will be proved in the sequel. The essence of our construction can be formulated simply as follows: it is worthwhile to combine information about matter field (contained in the “non-metricity tensor $N$” and given by formula \eqref{nonmetricity}) together with the metric connection $\mGamma$. The resulting non-metric connection $\Gamma = \ \mGamma + N$, significantly simplifies the structure of the theory.

Technically, the simplest way to prove the universal validity of Palatini formulation goes through the purely affine formulation. As will be shown in the next section, the affine picture arises naturally here and provides the easiest method to simplify the canonical structure of the theory.

\section{Derivation of the universal affine formulation of General Relativity Theory from the purely metric formulation }\label{univ-aff}

Let us, therefore, begin with the purely metric formulation of General Relativity Theory. In case of a generic matter Lagrangian density \eqref{ogo} we have:
\begin{eqnarray}
    \delta {\cal L}_{matt} (\phi ,\, \partial \phi ,\, g,\, \mGamma) &=&
    \frac {\partial {\cal L}_{matt}}{\partial g_{\mu\nu}}\,  \delta g_{\mu\nu} + {\cal P}_{\ \ \lambda}^{ \mu\nu}\, \delta
    \mGamma^{\lambda}_{\ \mu\nu} +
    \partial_\lambda \left( p^\lambda \, \delta \phi \right) \, ,\label{var Lag matt 0}
\end{eqnarray}
where, in general, the field ${\cal P}$ does not vanish identically:
\begin{equation}\label{def: calP}
    {\cal P}_{\ \ \lambda}^{\mu\nu} := \frac {\partial {\cal L}_{matt}}
    {\partial \mGamma^{\lambda}_{\ \mu\nu} } \neq 0 \, .
\end{equation}
\begin{lemma}
\label{lemma varGamma}
    The term $\mathcal{P}_{\ \ \lambda}^{ \mu\nu}\,  \delta\! \mGamma^{\lambda}_{\ \mu\nu}$ in formula \eqref{var Lag matt 0} can be rewritten as follows:
\begin{eqnarray*}
    \mathcal{P}_{\ \ \lambda}^{ \mu\nu}\,  \delta \mGamma^{\lambda}_{\  \mu\nu} &=&  \partial_{\kappa} \left( {\cal R}^{\mu\nu\kappa}\, \delta g_{\mu\nu} \right) - \left(\mnabla_{\kappa} \mathcal{R}^{\mu\nu\kappa}  \right)\,  \delta g_{\mu\nu}      \, ,
\end{eqnarray*}
where:
\begin{eqnarray}
    {\cal R}^{\mu\nu\kappa} &:=&  \frac 12 \left(  \mathcal{P}^{ \kappa\mu\nu} + \mathcal{P}^{\kappa \nu \mu} - \mathcal{P}^{ \mu\nu\kappa}  \right) \,
\label{def: calR}
\end{eqnarray}
and $\, \mnabla$ is the covariant derivative with respect to the Levi-Civita connection $\,\mGamma $.
\end{lemma}
The proof is given in {Appendix~\ref{proof lemma varGamma}}. Hence, formula \eqref{var Lag matt 0}, generating field dynamics, can be rewritten as:
\begin{eqnarray}
    \delta {\cal L}_{matt}  &=&\left(
    \frac {\partial {\cal L}_{matt}}{\partial g_{\mu\nu}} - \mnabla_{\kappa}\mathcal{R}^{\mu\nu\kappa} \right)\,  \delta g_{\mu\nu}  + \partial_{\kappa} \left( {\cal R}^{\mu\nu\kappa}\, \delta g_{\mu\nu} + p^{\kappa} \delta \phi \right) \, .
    \label{var Lag matt 10}
\end{eqnarray}

\

As mentioned already in the Introduction (cf.~formula \eqref{senm}), the derivative of the matter Lagrangian with respect to the metric in formula \eqref{var Lag matt 0}, has been replaced in the volume part of formula \eqref{var Lag matt 10} by its variational derivative:
\begin{eqnarray}
  \frac {\delta \Lag_{matt}(\phi , \partial \phi , g, \partial g)}{\delta g_{\mu\nu}} &=&
    \frac {\partial \Lag_{matt} }{\partial g_{\mu\nu}} -\partial_\kappa \frac {\partial \Lag_{matt} }{\partial g_{\mu\nu , \kappa}}
    =
    \frac {\partial {\cal L}_{matt} }{\partial g_{\mu\nu}} \ -
   \mnabla_{\kappa}\mathcal{R}^{\mu\nu\kappa}
   \label{comp}
  \, .
\end{eqnarray}
The price for this change is the new term $\partial_{\kappa} \left( {\cal R}^{\mu\nu\kappa}\, \delta g_{\mu\nu} \right)$ arising in the boundary part of the formula. Adding now the variation of the Hilbert Lagrangian given by \eqref{varG},  we obtain the following, universal formula for the variation of the total, metric Lagrangian density ${\cal L}_g= {\cal L}_H + {\cal L}_{matt}$:
\begin{eqnarray}\label{Lg_gen_u}
    \delta {\cal L}_{g} (\phi ,\, \partial \phi ,\, g,\, \partial g,\, \partial^2 g)&=&  \left[\frac {\partial {\cal L}_{matt}}{\partial g_{\mu\nu}}  -\frac{1}{16\pi}\stackrel{\circ \ \ }{\mathcal{G}^{\mu\nu}} -
   \mnabla_{\kappa}\mathcal{R}^{\mu\nu\kappa} \right] \delta g_{\mu\nu} +
   \nonumber \\
   && + \partial_{\kappa} \left( {\cal R}^{\mu\nu\kappa}\, \delta g_{\mu\nu} +  p^{\kappa} \, \delta \phi  +
    {\pi}_{\lambda}^{\ \mu\nu\kappa} \delta
    \mGamma^{\lambda}_{\ \mu\nu} \right) \, ,
\end{eqnarray}
(we have put the circle above the Einstein tensor density $\kolo{\mathcal G}$, in order to stress that it is calculated for the Levi-Civita metric connection $\mGamma$.)

Above formula splits into the “volume part”, describing  Euler-Lagrange field equations (equivalent to the standard metric Einstein equations):
\begin{eqnarray}\label{E-nm}
 \frac {\delta {\cal L}_{matt}}{\delta g_{\mu\nu}}=
 \frac {\partial {\cal L}_{matt}}{\partial g_{\mu\nu}}  -\frac{1}{16\pi}\stackrel{\circ \ \ }{\mathcal{G}^{\mu\nu}} -
   \mnabla_{\kappa}\mathcal{R}^{\mu\nu\kappa}    =0 \, ,
\end{eqnarray}
and the “boundary part”, describing the “\textit{on shell}" (i.e. when the field equations {\em are satisfied}) -- variation of $\Lag_g$:
\begin{eqnarray} \label{var Lag g 1}
    \delta {\cal L}_{g}   &=&   \partial_{\kappa} \left( {\cal R}^{\mu\nu\kappa}\, \delta g_{\mu\nu} +  p^{\kappa} \, \delta \phi  +
    {\pi}_{\lambda}^{\ \mu\nu\kappa}\,\delta
    \mGamma^{\lambda}_{\ \mu\nu} \right) \, .
\end{eqnarray}

We show in Appendix \ref{proof R} that the troublesome term   $\mnabla_{\kappa}\mathcal{R}^{\mu\nu\kappa}$ combines, together with the metric Einstein-tensor-density $\stackrel{\circ \ \ }{\mathcal{G}^{\mu\nu}}$, to its non-metric analogue ${\mathcal{G}^{\mu\nu}}$ (plus an extra term, which also finds a nice geometric interpretation and is discussed in the sequel).

\

Unfortunately, above formulae are not yet fully satisfactory because the metric $g$ appears here in a double role: as a control parameter ($\delta g_{\mu\nu})$ and as the response parameter ($\pi_{\lambda}^{\ \mu\nu\kappa} $). Mathematically, this means that the underlying symplectic structure is degenerate: some combinations of the “momenta” $\pi^{\mu\nu}$ (response) are equal to some combinations of the “configuration” $g_{\mu\nu}$ (control). Now, we are going to reduce this degeneracy, rewriting the formula in terms of independent parameters. First, we rewrite the term $\partial_{\kappa} \left( {\cal R}^{\mu\nu\kappa}\, \delta g_{\mu\nu}\right)$ according to the following lemma:
\begin{lemma}
\label{lemma non-metricity}
The following identity holds:
\begin{equation}\label{tozsam}
    \partial_{\kappa} \left( {\cal R}^{\mu\nu\kappa}\, \delta g_{\mu\nu} \right)  =\partial_{\nu}
    \left[ \pi_{\kappa}^{\ \lambda \mu\nu}\,  \delta  N^{\kappa}_{\ \lambda \mu}  \right]  + \delta \left[ \mnabla_\kappa  \mathcal{R}_{\sigma}^{\ \sigma\kappa} \right] \, ,
\end{equation}
where
\begin{eqnarray}
    N^{\kappa}_{\ \lambda \mu} &:= &\frac{16 \pi}{\sqrt{|\det g|}}\, \left[{\cal R}_{\lambda \mu}^{\ \ \ \kappa} -    \frac 12\,  {\cal R}_{\sigma}^{\ \ \sigma \kappa} \, g_{\lambda \mu}  - \frac 23\,  \left(\delta^\kappa_{(\lambda }\,{\cal R}_{\mu)\sigma}^{\ \ \ \sigma} - \frac 12\,{\cal R}^{\sigma}_{\ \sigma(\lambda}\, \delta^\kappa_{\mu)} \right) \right] \, , \ \
\label{def: tensor N obraz metryczny}
\end{eqnarray}
is simply equal to the non-metricity \eqref{nonmetricity}, rewritten in terms of the auxiliary quantity ${\cal R}$ defined by \eqref{def: calR}, instead of the original object~${\cal P}$.
\end{lemma}
Proof of this Lemma is presented in Appendix \ref{proof lemma non-metricity}. Using this identity, we rewrite the variation \eqref{var Lag g 1} of the metric Lagrangian ${\cal L}_g$ as follows:
\begin{eqnarray}
    \delta {\cal L}_{g}   &=&   \partial_{\kappa} \left[  p^{\kappa} \, \delta \phi  +
    {\pi}_{\lambda}^{\ \mu\nu\kappa} \delta \left(
    \mGamma^{\lambda}_{\ \mu\nu}+N^{\lambda}_{\ \mu\nu} \right)\right]+  \delta \left[ \mnabla_\kappa  \mathcal{R}_{\sigma}^{\ \sigma\kappa} \right] \,. \label{var Lag g 2}
\end{eqnarray}
We see, that our non-metric connection $\Gamma =\ \mGamma+N$, already defined in \eqref{Gamma-gen}, arises here naturally:
\begin{eqnarray} \label{var Lag g 3}
    \delta {\cal L}_{g} &=&   \partial_{\kappa} \left[  p^{\kappa} \, \delta \phi  +
    {\pi}_{\lambda}^{\ \mu\nu\kappa} \, \delta \Gamma^{\lambda}_{\ \mu\nu} \right]
    + \delta \left[ \mnabla_\kappa  \mathcal{R}_{\sigma}^{\ \sigma\kappa} \right]\, .
\end{eqnarray}

Now, we perform the Legendre transformation between the metric and the connection. As the first step, we put the last term (the complete variation) on the left-hand side and obtain finally the universal affine Lagrangian:
\begin{equation}\label{Pala-def}
    \Lag_A:=\Lag_g -  \mnabla_\kappa  \mathcal{R}_{\sigma}^{\ \sigma\kappa}   \, ,
\end{equation}
where now only the symmetric (but not necessarily metric!) connection $\Gamma$ and the matter fields $\phi$ play the role of independent configuration variables:
\begin{eqnarray} \label{var Lag P 1}
    \delta {\cal L}_{A}  &=&   \partial_{\kappa} \left[  p^{\kappa} \, \delta \phi  +
    {\pi}_{\lambda}^{\ \mu\nu\kappa}\, \delta \Gamma^{\lambda}_{\ \mu\nu} \right]= \\
    &=& \left(\partial_\kappa p^\kappa \right) \, \delta \phi  +
    p^\kappa\, \delta \phi_\kappa + \left( \partial_\kappa {\pi}_{\lambda}^{\ \mu\nu\kappa}
    \right)\, \delta \Gamma^{\lambda}_{\ \mu\nu} + \pi_{\lambda}^{\ \mu\nu\kappa}\,  \delta
    \Gamma^{\lambda}_{\ \mu\nu,\kappa}
    \, . \label{var Lag P 12}
\end{eqnarray}
The metric $g_{\mu\nu}$ (represented here by the momentum $\pi^{\mu\nu}$) and its derivatives (represented by $\partial_\kappa {\pi}_{\lambda}^{\ \mu\nu\kappa}$) are now shifted to the level of “response parameters”: they arise as derivatives of the new affine Lagrangian with respect to $\Gamma^{\lambda}_{\ \mu\nu}$ and its derivatives $\Gamma^{\lambda}_{\ \mu\nu,\kappa}$. But, to be consistent, the second step of the Legendre transformation must follow the first one: the function ${\cal L}_A$ must be expressed in terms of the new control parameters, whereas the old control parameters must be eliminated with the help of the appropriate field equations. We have an analogous situation in thermodynamics where, to perform transition from the adiabatic to the thermostatic mode, it is not sufficient to replace formula $\delta U(V, S) = -p \delta V + T \delta S$ by $\delta \left( U - TS \right) =- p\delta V - S \delta T$, but the new generating function $H := U - TS$ (the Helmholtz “free energy”) must be expressed in terms of the new control parameters: $H = H(V,T)$.
A similar phenomenon occurs during any Legendre transformation, which we commonly use in theoretical physics. For example, formula
\begin{equation}\label{Ham}
    H(p,q,\dot{q}) = p  \, \dot{q} - L(q,\dot{q})\, , \qquad
    p:= \frac{\partial L}{\partial \dot{q}} \, ,
\end{equation}
describing transformation from the Lagrangian control mode to the Hamiltonian one in classical mechanics means, that we have to eliminate the Lagrangian variable $\dot{q}$ and express it in terms of the “true Hamiltonian control parameters” $(p,q)$.
Mathematically rigorous definition of these geometric structures can be found in \cite{Tulcz}, \cite{Kij-Moreno2015}. We stress, however, that no assumption about the non-degeneracy of the original Lagrangian density is necessary here: theories with constraints, where only a subspace of control parameters is physically accessible and, consequently, the “control -- response” relation is no longer equivocal, (i.e.: non-physical “gauge parameters” arise), fit perfectly into this framework.

\

To perform the second step of the Legendre transformation between the metric and the affine picture in the gravity theory, we observe that the numerical value of the affine Lagrangian ${\cal L}_A$ equals:
\begin{eqnarray}
    \Lag_A&:=& \Lag_g -  \mnabla_\kappa  \mathcal{R}_{\sigma}^{\ \sigma\kappa}   =
    \Lag_{matt}(\phi , \partial \phi , g,\mGamma) + \Lag_{H}(g_{\mu\nu},\kolo{R}_{\mu\nu}) - \mnabla_\kappa  \mathcal{R}_{\sigma}^{\ \sigma\kappa}   \, . \label{LA-1}
\end{eqnarray}
Hence, we have {\em a priori}:
\begin{equation}\label{LA-2}
    \Lag_A = \Lag_A (\phi ,\, \partial \phi ,\, g_{\mu\nu},\, \partial_\kappa g_{\mu\nu},\,  \Gamma^{\lambda}_{\ \mu\nu},\, R_{\mu\nu}) \, .
\end{equation}
But now, the variables $(g_{\mu\nu},\, \partial_\kappa g_{\mu\nu})$ are no longer “control parameters” -- see variational formula \eqref{var Lag P 1} -- and have to be eliminated by virtue of dynamical equations. For this purpose, we can use two equations which enable us to reduce the number of control parameters: the non-metricity equation \eqref{non-m} and the Einstein equation \eqref{E-nm}:
\begin{eqnarray}\label{aaa}
 \nabla_\kappa {\pi}_{\lambda}^{\ \mu\nu\kappa}
    &=&
    {\cal P}_{\ \ \lambda}^{\mu\nu} := \frac {\partial {\cal L}_{matt}}
    {\partial \mGamma^{\lambda}_{\ \mu\nu} }\, ,  \label{Leg-m-a1}\\
    \frac{1}{16\pi}\stackrel{\circ \ \ }{\mathcal{G}^{\mu\nu}} &=&\frac {\partial {\cal L}_{matt}}{\partial g_{\mu\nu}}
    - \mnabla_{\kappa}\mathcal{R}^{\mu\nu\kappa}\, .\label{Leg-m-a2}
\end{eqnarray}

The complete Legendre transformation from the metric picture (“control mode”) to the affine picture is, therefore, given by definition \eqref{LA-1} of the affine Lagrangian density, together with formulas \eqref{Leg-m-a1} and \eqref{Leg-m-a2}, the latter being necessary to express “response parameters” $(g ,\, \mGamma )$ in terms of the new control parameters: ($\Gamma,\, \partial\Gamma$).

As will be seen in the Appendix E, derivatives of the connection enter here {\em via} the Ricci tensor $R_{\mu\nu}$ only:
\begin{eqnarray}
    R_{\mu \nu} :=  -\Gamma^{\kappa}_{\ \kappa\mu, \nu} + \Gamma^{\kappa}_{\ \mu \nu, \kappa}-\Gamma^{\sigma}_{\  \kappa\mu}\, \Gamma^{\kappa}_{\ \nu \sigma} + \Gamma^{\sigma}_{\ \mu \nu}\, \Gamma^{\kappa}_{\ \kappa \sigma}\, . \quad
    \label{def: tensor Ricciego}
\end{eqnarray}
and even less, namely {\em via} its symmetric part (Ricci tensor of a non-metric connection may contain a skew-symmetric part!):
\begin{eqnarray}
 K_{\mu\nu}&:=& R_{(\mu\nu)}=-\Gamma^{\kappa}_{\ \kappa(\mu, \nu)}+\Gamma^{\kappa}_{\ \mu  \nu, \kappa}  -\Gamma^{\sigma}_{\  \kappa\mu}\, \Gamma^{\kappa}_{\ \nu \sigma} + \Gamma^{\sigma}_{\ \mu \nu}\, \Gamma^{\kappa}_{\ \kappa \sigma}\, ,
    \label{def: K}
\end{eqnarray}
similarly as it was the case in the metric picture. The formula expressing $K_{\mu\nu}$ in terms of its metric analogue $\kolo{R}_{\mu\nu}$ is given in Appendix \ref{proof R}.

Practically, to solve equations \eqref{Leg-m-a1}, \eqref{Leg-m-a2}, and to express explicitly metric $g$ and the metric connection $\mGamma$ in terms of $K$ and $\Gamma$ (together with the matter field $\phi$) can be computationally difficult (see example in Appendix \ref{example}). But, theoretically, the implicit relation between the metric and the affine control parameters is sufficient. It might well be that “what is simple in the metric picture, is complicated in the affine one and {\em vice versa}”. The affine picture being conceptually much simpler that the metric picture, we are deeply convinced that when looking for fundamental laws of nature (e.g.~description of the dark matter) one should begin with geometric structures which are simple rather in affine picture, even if their metric counterparts would be computationally complicated. From this point of view, attempts to describe dark matter by modifying slightly the metric Hilbert Lagrangian density do not look very convincing. However, there are important examples of the matter fields, where the above Legendre transformation can be performed explicitly. For the reader's convenience, we list below a few examples:
\begin{itemize}
  \item
  \begin{equation}\label{cosm}
    \Lag_A(K_{\mu\nu}) = C \cdot \sqrt{|\det K_{\mu\nu}|} \, .
  \end{equation}
  The resulting field equations are the Einstein equations for empty space with cosmological constant $\Lambda:= \frac 1{8\pi\,  C}$.
    \item
  \begin{eqnarray}
 \Lag_A(K_{\mu\nu}, \phi, \partial_{\alpha} \phi) &=&  \frac{2}{m^2 \phi^2} \sqrt{\left|\det
\left(\frac{1}{8\pi} K_{\mu\nu} - \phi_{,\mu} \phi_{,\nu}\right)\right|} \, , \label{eq:exlaff}
  \end{eqnarray}
describes the Klein-Gordon-Einstein theory of the self-gravitating scalar field $\phi$ with mass $m$.
  \item
  \begin{eqnarray}
    \Lag_A(K_{\mu\nu},\partial_{\alpha} A_{\beta}) &= & - \frac 14\, \sqrt{|\det K_{\rho\sigma}|}\, (K^{-1})^{\mu\nu}\, (K^{-1})^{\alpha\beta}\, F_{\mu\alpha}\, F_{\nu\beta} \, , \label{eq:E-M}
  \end{eqnarray}
where $F_{\mu\nu} := \partial_\mu A_\nu - \partial_\nu A_\mu$ is the electromagnetic field (Faraday tensor) given in terms of its four-potential $A_\mu$, describes the Maxwell-Einstein theory.

\end{itemize}

For more examples see, e.g.,  \cite{kij-magli1997}, \cite{kij-magli1998}, \cite{kij-werp2007}.

\section{Transition to the universal Palatini picture}\label{univ-P}

In this section, we introduce the universal “Palatini picture”, where the configuration fields are: the metric tensor $g$ and the symmetric connection $\Gamma$. In contrast to the “naive” Palatini picture, discussed in Section 3, no metricity condition for $\Gamma$ is assumed {\em a priori}. The non-metricity equation \eqref{Leg-m-a1} and the  Einstein equation \eqref{Leg-m-a2} arise here as Euler-Lagrange field equations, implied by variation with respect to $\Gamma$ and to $g$, respectively.

This universal Palatini formulation can be obtained directly {\em via} a simple  Legendre transformation from the affine picture. Let us begin, therefore, with a generic affine Lagrangian
\begin{equation}\label{LA-10}
    \Lag_A = \Lag_A (\phi ,\,  \partial \phi ,\,  \Gamma^{\kappa}_{\ \lambda\mu},\, \Gamma^{\kappa}_{\ \lambda\mu,\nu})
    = \Lag_A(\phi ,\,  \nabla \phi ,\,  K_{\mu\nu})\, ,
\end{equation}
where derivatives of the connection coefficients enter {\em via} the symmetric part
\eqref{def: K} of the Ricci, exclusively. Let us introduce the “coordinate transformation”:
\begin{equation*}
    (\Gamma,\, \partial\Gamma) \mapsto (\Gamma,\, K)
\end{equation*}
in space of control parameters.
\begin{lemma}
\label{lemma5.1}
The following identity holds:
\begin{eqnarray}
 \partial_{\kappa}\left(\pi_{\kappa}^{\ \lambda\mu\nu}\, \delta \Gamma^{\kappa}_{\ \lambda\mu} \right) = \left( \nabla_\kappa {\pi}_{\lambda}^{\ \mu\nu\kappa}
    \right)\, \delta \Gamma^{\lambda}_{\ \mu\nu} + \pi^{\mu\nu} \, \delta
    K_{\mu\nu}\, . \ \
\end{eqnarray}
\end{lemma}
The proof of this  identity is given in Appendix \ref{proof th 6.2}. It enables us to rewrite generating formula \eqref{var Lag P 1} as follows:
\begin{eqnarray}\label{var Lag P 2-}
    \delta {\cal L}_{A}   &=&   \partial_{\kappa} \left(  p^{\kappa} \, \delta \phi  \right) + \left(\nabla_{\kappa}{\pi}_{\lambda}^{\ \mu\nu\kappa}\right)\, \delta \Gamma^{\lambda}_{\ \mu\nu}+ \pi^{\mu\nu}\, \delta K_{\mu\nu}=  \\
    &=&
    \partial_{\kappa} \left(  p^{\kappa} \, \delta \phi  \right)  + \left(\nabla_{\kappa}{\pi}_{\lambda}^{\ \mu\nu\kappa}\right)\, \delta \Gamma^{\lambda}_{\ \mu\nu}  - K_{\mu\nu}\,  \delta \pi^{\mu\nu}  + \delta \left( \pi^{\mu\nu}\,  K_{\mu\nu} \right)
    \,.
    \label{var Lag P 2}
\end{eqnarray}
Putting the complete variation $\delta \left( \pi^{\mu\nu}\,  K_{\mu\nu} \right)$ on the left-hand side, we obtain the Universal Palatini Lagrangian:
\begin{eqnarray}
    {\cal L}_P (\phi,\, \partial\phi,\, \Gamma^\lambda_{\ \mu\nu},
     \Gamma^{\kappa}_{\ \lambda\mu,\nu}, g_{\mu\nu}) &:=& {\cal L}_{A} - \pi^{\mu\nu}\, K_{\mu\nu} =  {\cal L}_{A} - \frac {\sqrt{|\det g|}}{16 \pi}\, R\, , \label{Pal8}
\end{eqnarray}
where the subtracted term is equal to the scalar curvature $\pi^{\mu\nu}\,  K_{\mu\nu} = \pi^{\mu\nu}\,  R_{\mu\nu}$ of a general (possibly non-metric) connection, analogous to the Hilbert Lagrangian $\Lag_H(g,\, \Gamma,\, \partial \Gamma)$ in formula~\eqref{Lag H}. Due to \eqref{var Lag P 2}, this Lagrangian  generates field equations according to:
\begin{eqnarray}
    \delta {\cal L}_{P}
    &=& \left(\nabla_{\kappa}{\pi}_{\lambda}^{\ \mu\nu\kappa}\right)\, \delta \Gamma^{\lambda}_{\ \mu\nu}  + \frac 1{16 \pi}\, {\cal G}^{\mu\nu}\,  \delta g_{\mu\nu} +  \partial_{\kappa} \left(  p^{\kappa} \, \delta \phi  \right)  \, , \nonumber  \\
    \label{var Lag P 23}
\end{eqnarray}
(cf.~identity \eqref{R delta pi}) or, equivalently:
\begin{eqnarray}
 \frac{1}{16\pi}\,  {\cal G}^{\mu\nu} = \frac{\partial \Lag_P}{\partial g_{\mu\nu}}\, ,\qquad \qquad
 \nabla_{\kappa}{\pi}_{\lambda}^{\ \mu\nu\kappa}= \frac{\partial \Lag_P}{\partial \Gamma^{\lambda}_{\ \mu\nu}}\, ,
\end{eqnarray}
(cf.~\eqref{Leg-m-a1} and \eqref{Leg-m-a2}).

In case of a special theory \eqref{Lmatt-specific}, where matter Lagrangian density does not depend upon derivatives of the metric tensor (i.e.~upon the metric connection $\mGamma$), the non-metricity of $\Gamma$ vanishes (see Lemmas \ref{lemma varGamma}  and \ref{lemma non-metricity}) and we have $K_{\mu\nu}=\ \kolo{R}_{\mu\nu}$ (cf. Appendix \ref{proof R}). Consequently, definition \eqref{Pal8} implies in that case:
\begin{eqnarray}
    {\cal L}_P &=& {\cal L}_{A} - \pi^{\mu\nu} K_{\mu\nu} = {\cal L}_{matt} + {\cal L}_H(g,\, \partial g,\, \partial^2 g) - \pi^{\mu\nu} \kolo{R}_{\mu\nu} = {\cal L}_{matt} \, ,\label{Palatini-A}
\end{eqnarray}
i.e.~our “universal Palatini picture” reduces to the classical (“naive”) Palatini method of variation (cf. \cite{palatini}), discussed in Section 3. But, in a generic case \eqref{ogo}, the two pictures differ from each other. In particular, we have:
\begin{eqnarray}
    \Lag_P &=&
    \Lag_{matt} -\pi^{\mu\nu}\,\left( N^{\sigma}_{\ \mu \nu}\, N^{\kappa}_{\ \kappa \sigma} - N^{\sigma}_{\ \kappa \mu}\, N^{\kappa}_{\ \nu \sigma}\right)= \nonumber \\
    &=&\Lag_{matt} -
    \frac{\sqrt{|\det g|}}{16 \pi}\, g^{\mu\nu}\,\left( N^{\sigma}_{\ \mu \nu}\, N^{\kappa}_{\ \kappa \sigma} - N^{\sigma}_{\ \kappa \mu}\, N^{\kappa}_{\ \nu \sigma}\right)
    \, ,
    \label{LP-1}
\end{eqnarray}
where formula \eqref{LP-1} is implied by \eqref{LA-1}, \eqref{Pal8} and the following identity, which is an easy consequence of the relation between $K_{\mu\nu}$ and $\kolo{R}_{\mu\nu}$ (see Appendix \ref{proof R}):
\begin{eqnarray}
 \pi^{\mu\nu} \kolo{R}_{\mu\nu} - \pi^{\mu\nu} \, K_{\mu\nu} - \mnabla_\kappa  \mathcal{R}_{\sigma}^{\ \sigma\kappa} = -\pi^{\mu\nu}\left( N^{\sigma}_{\ \mu \nu}\, N^{\kappa}_{\ \kappa \sigma} - N^{\sigma}_{\ \kappa \mu}\, N^{\kappa}_{\ \nu \sigma}\right)\, .
\end{eqnarray}
We see explicitly that:
\begin{equation}\label{KO}
    \frac {\partial {\cal L}_{matt}}{\partial \Gamma^\lambda_{\ \mu\nu}} = 0 \    \Longleftrightarrow
    \    {\cal P}_{\ \ \lambda}^{\mu\nu}= 0  \   \Longleftrightarrow   \   N^\lambda_{\ \mu\nu}= 0 \
    \Longrightarrow \   \Lag_P = \Lag_{matt}\, .
\end{equation}

\subsection{Fast track from the metric picture to the universal Palatini picture}

Formula \eqref{LP-1} provides the method of a direct transition from the metric picture to the universal Palatini picture, without transition {\em via} the affine picture.  The formula gives directly the value of the universal Palatini Lagrangian density $\Lag_P$ as the function of the following variables:
\begin{eqnarray}
   \Lag_P =\Lag_P(g, \phi, \nabla \phi, \Gamma ,N) = \Lag_{matt}(g, \phi, \mnabla \phi, \mGamma) -\pi^{\mu\nu}\,\left( N^{\sigma}_{\ \mu \nu}\, N^{\kappa}_{\ \kappa \sigma} - N^{\sigma}_{\ \kappa \mu}\, N^{\kappa}_{\ \nu \sigma}\right)\, ,\label{lagp3}
\end{eqnarray}
where  $\  \mGamma := \Gamma - N$”.
What remains now is to eliminate the “illegal” variable $N$, which must be expressed in terms of the remaining (“legal”) variables. For this purpose we have to solve (with respect to $N$) equation~\eqref{N=lincalP}, with ${\cal P}$ given by  \eqref{Leg-m-a1}, i.e.~by the definition of non-metricity:
\begin{equation*}
    {\cal P}_{\ \ \lambda}^{\mu\nu} = \frac {\partial {\cal L}_{matt}}
    {\partial \mGamma^{\lambda}_{\ \mu\nu} }\, .
\end{equation*}
{\bf Example:}
As the “metric matter-Lagrangian-density” of the vector field $\phi^{\alpha}$, take the standard, quadratic form of its covariant derivatives:
\begin{eqnarray*}
 \Lag_{matt}(g,\, \mnabla \phi) =\frac{\sqrt{|\det g|} }{16 \pi} \,\left( \mnabla_{\alpha}\phi^{\beta}\right)\, \left( \mnabla^{\alpha}\phi_{\beta}\right)\, .
\end{eqnarray*}
We prove in Appendix \ref{example} that the corresponding Palatini Lagrangian \eqref{LP-1} is given by the following formula:
\begin{eqnarray}
 \Lag_P(g,\nabla \phi)= \frac{\sqrt{|\det g|} }{32\pi\, D_1\, D_2\,
(\phi^2-1)}\! && \bigg\{ D_1\,  D_2\, \left(\phi^2-2\right) \,
(\nabla_{\alpha}\phi^{\mu})\,(\nabla_{\beta}\phi^{\nu}) \,
g^{\alpha\beta}\, g_{\mu\nu}    + \nonumber\\
&&  + D_1\,  D_2\,   \phi^2  \, (\nabla_{\alpha}\phi^{\mu})\,
(\nabla_{\mu}\phi^{\alpha})  + \nonumber \\
&&+4\phi^2\, \left(\phi^2-1\right)\,
\left[\phi^{\alpha}\, \left(\nabla_{\alpha} \phi^{\mu}\right)\,
\phi_{\mu} \right]^2 + \nonumber\\
&&+ 4 D_2\,  \left(\phi^2-1\right)\, \left[\phi^{\alpha}\,
\left(\nabla_{\alpha} \phi^{\mu}\right)\, \phi_{\mu} \right]\,
\left(\nabla_{\beta} \phi^{\beta}\right) + \nonumber \\
&&-  \phi^2\,  D_2\,
\left(\phi^2-1\right)\, \left(\nabla_{\beta} \phi^{\beta}\right)^2 +\nonumber \\
  & & - D_1\,  \left(\phi^2-2\right)^2 \, \left[ (\nabla_{\mu}
\phi^{\alpha})\, \phi_{\alpha}\right]  \, \left[ (\nabla_{\nu}
\phi^{\beta})\, \phi_{\beta}\right]\, g^{\mu\nu}  + \nonumber\\
&& - 2 D_1\,      \left(\phi^2-2\right) \, \phi^2\,  \left[
(\nabla_{\mu} \phi^{\alpha})\, \phi_{\alpha}\right]  \, \left[
\phi^{\beta}\, (\nabla_{\beta} \phi^{\mu})\right] +\nonumber \\
  & &    - D_1\, \phi^4\,\left[  \phi^{\alpha}\, (\nabla_{\alpha}
\phi^{\mu})\right]\,  \left[  \phi^{\beta}\, (\nabla_{\beta}
\phi^{\nu})\right]\, g_{\mu\nu} \bigg\}
\, ,
\end{eqnarray}
where $D_1$, $D_2$ are polynomials of the variable
\begin{equation*}
\phi^2:=~\phi_{\alpha}\,\phi^{\alpha}\, ,
\end{equation*}
defined in Appendix \ref{example}.

\textbf{Example:} Transformation in the opposite direction. As the “Palatini-Lagrangian density” of the vector field $X^{\alpha}$, take the simplest, quadratic form of its covariant (with respect to the non-metric connection $\Gamma$) derivatives:

\begin{eqnarray*}
 \Lag_P(g,\, \nabla X)=\frac{\sqrt{|\det g|} }{16\pi} \,\left( \nabla_{\alpha} X^{\mu}\right)\, \left( \nabla_{\beta}X^{\nu}\right)\, g^{\alpha\beta}\, g_{\mu\nu} \, .
\end{eqnarray*}
We prove in Appendix \ref{example} that the corresponding metric matter Lagrangian density is equal to:
\begin{eqnarray}
 \Lag_{matt}(g,\, \mnabla X) = \frac{\sqrt{|\det g|} }{32\pi\, \widetilde{D}_1\, \widetilde{D}_2\, (X^2+1)} && \bigg\{  \widetilde{D}_1 \,  \widetilde{D}_2 \,\left(X^2+2\right) \, \left(\mnabla_{\alpha}X_{\beta}\right)\left(\mnabla^{\alpha}X^{\beta}\right) +   \nonumber \\
 &&+  \widetilde{D}_1 \,  \widetilde{D}_2 \,X^2\,  \left(\mnabla_{\alpha}X_{\beta}\right)\left(\mnabla^{\beta}X^{\alpha}\right)  +   \nonumber \\
 && + 4 \left(X^2+1\right)\, X^2\, \left[X^{\alpha}\, \left(\mnabla_{\alpha} X^{\mu}\right)\, X_{\mu} \right]^2 +   \nonumber \\
 && +  4 \widetilde{D}_2\,  \left(X^2+1\right)\,  \left[X^{\alpha}\, \left(\mnabla_{\alpha} X^{\mu}\right)\, X_{\mu} \right]\, \left(\mnabla_{\beta} X^{\beta}\right)+\nonumber \\
 &&-  X^2\,    \widetilde{D}_2 \,  \left(X^2+1\right)\,  \left(\mnabla_{\alpha} X^{\alpha}\right)^2  +   \nonumber \\
 && - \widetilde{D}_1\, \left(X^2-2\right)^2\,    \left[ (\mnabla_{\mu} X^{\alpha})\, X_{\alpha}\right]  \, \left[ (\mnabla^{\mu} X^{\beta})\, X_{\beta}\right] +\nonumber \\
 &&- 2  \widetilde{D}_1 \, \left(X^2+2\right)\,  X^2\, \left[
(\mnabla_{\mu} X^{\alpha})\, X_{\alpha}\right]  \, \left[  X^{\beta}\,
(\mnabla_{\beta} X^{\mu})\right]  +   \nonumber \\
 && - \widetilde{D}_1   \,X^4\, \left[  X^{\alpha}\, (\mnabla_{\alpha} X^{\mu})\right]\,  \left[  X^{\beta}\, (\mnabla_{\beta} X_{\mu})\right]\bigg\}
\, ,
\end{eqnarray}
where $\widetilde{D}_1$ i $\widetilde{D}_2$ are polynomials of the variable
\begin{equation*}
X^2:= X_{\alpha}\,X^{\alpha}\, ,
\end{equation*}
defined in Appendix \ref{example2}.

The two examples illustrate nicely our conclusion: “what is simple in the metric picture can be computationally difficult in the affine picture and {\em vice versa}”. Being conceptually much simpler (as far as the canonical structure of the theory is concerned, cf. \cite{pieszy}), the affine formulation (together with its consequence: the universal Palatini picture) are -- in our opinion -- better suited to become the starting point for any attempt to extrapolate validity of the present Gravity Theory by many orders of magnitude, namely from our Solar System scale to the cosmological scale, i.e.~to the dark matter scale (cf.~discussion in \cite{uni}).

\section{From affine to metric picture} \label{afin}

The equivalence proof of the three different formulations of the Gravity Theory (purely metric, metric-affine which we call “Palatini” and -- finally -- purely affine) would not be complete without a brief discussion of what happens if we begin with a purely affine Lagrangian density. Even if these results are contained {\em implicite} in our previous sections, we have decided, for the convenience of the reader, to give below the short equivalence proof starting from the affine picture. The main benefit of this discussion is that, from an affine perspective, it is obvious that the current Theory of Gravity is only a narrow sector of a much broader theory. This observation gives room for the description of other physical fields within a uniform geometric structure (cf. \cite{ein-autob} and \cite{uni}).

Consider, therefore, a generic affine variational principle, where the field configuration is described by the first jet of a symmetric connection $(\Gamma,\, \partial\Gamma)$ and the first jet of a matter field $(\phi, \partial \phi)$:
\begin{eqnarray}
 \delta\Lag_A (\Gamma,\, \partial\Gamma,\, \phi, \partial \phi)  &=&  \partial_{\nu}\left(\mathcal{P}_{\kappa}^{\ \lambda\mu\nu}\, \delta\Gamma^{\kappa}_{\ \lambda\mu} + p^{\nu}\, \delta\phi \right)=\nonumber   \\
&=& \left(\partial_\nu \mathcal{P}_{\kappa}^{\ \lambda\mu\nu}\right) \delta\Gamma^{\kappa}_{\ \lambda\mu}+
\mathcal{P}_{\kappa}^{\ \lambda\mu\nu}\, \delta\Gamma^{\kappa}_{\ \lambda\mu,\nu} + \nonumber \\
& &  +
\left( \partial_\nu p^\nu \right) \delta \phi + p^\nu \delta \phi_{,\nu} \, ,\label{varaffa}
\end{eqnarray}
or, equivalently,
\begin{eqnarray}
    \left\{
    \partial_\nu \mathcal{P}_{\kappa}^{\ \lambda\mu\nu} = \frac {\partial \Lag_A}{\partial \Gamma^{\kappa}_{\ \lambda\mu}}\, ,
    \ \mathcal{P}_{\kappa}^{\ \lambda\mu\nu} = \frac {\partial \Lag_A}{\partial \Gamma^{\kappa}_{\ \lambda\mu,\nu}} \right\}  \
   &\Longrightarrow &   \ \frac {\delta \Lag_A}{\delta \Gamma^{\kappa}_{\ \lambda\mu}} = 0\, ,  \label{DefP}
       \\
       \left\{
   \partial_\nu p^\nu = \frac {\partial \Lag_A}{\partial \phi}
     \, , \   p^\nu = \frac {\partial \Lag_A}{\partial \phi_{,\nu}} \right\} \   &\Longrightarrow &
     \   \frac {\delta \Lag_A}{\delta \phi} = 0
     \, .
\end{eqnarray}
As explained in Section \ref{sec var}, formula \eqref{varaffa} is equivalent to Euler-Lagrange equations, together with definition of momenta: $\mathcal{P}_{\kappa}^{\ \lambda\mu\nu}$ and $p^\nu$, canonically conjugate to $\Gamma^{\kappa}_{\ \lambda\mu}$ and $\phi$, respectively.

To manufacture an invariant scalar-density ${\cal L}_A$ from derivatives of $\Gamma$, the only method is to use the unique tensorial object which can be constructed from them, namely the Riemann tensor \eqref{def riemann}:
\begin{equation*}
    R^{\kappa}_{\ \lambda\mu\nu} = \Gamma^{\kappa}_{\ \lambda \nu, \mu} -\Gamma^{\kappa}_{\ \lambda \mu, \nu}+
    \Gamma^{\kappa}_{\ \sigma\mu } \, \Gamma^{\sigma}_{\ \lambda \nu}
    -\Gamma^{\kappa}_{\ \sigma\nu } \, \Gamma^{\sigma}_{\ \lambda \mu} \, .
\end{equation*}
Moreover, the connection $\Gamma$ can only enter into the game either {\em via} the curvature $R^{\kappa}_{\ \lambda\mu\nu}$ itself, or {\em via} some “covariant derivatives” of the matter field, which can be symbolically written as:
\begin{equation}\label{covphi}
    \nabla \phi := \partial \phi + ``\ \Gamma \cdot \phi \ " \, .
\end{equation}
This means that, in fact, we have $\Lag_A = \Lag_A (R^{\kappa}_{\ \lambda\mu\nu} , \phi , \nabla \phi)$. The Riemann tensor satisfies identities:
\begin{eqnarray*}
    R^{\kappa}_{\ \lambda\mu\nu} &=& - R^{\kappa}_{\ \lambda\nu\mu} \, \text{(skew-symmetry)}\, , \\
    R^{\kappa}_{\ [\lambda\mu\nu]} &=& 0\, \text{ (Bianchi I identity)}\, ,
\end{eqnarray*}
and, consequently, splits into three irreducible parts. First, its splits into its trace --  the Ricci tensor \eqref{def ricci} $R_{\lambda\nu}:=R^{\kappa}_{\ \lambda\kappa\nu}$ -- and the remaining traceless part $W^{\kappa}_{\ \lambda\mu\nu}$. The latter is an analogue of the Weyl tensor in the metric case. It is defined by the following identities:
\[
    W^{\kappa}_{\ \kappa\mu\nu} = W^{\kappa}_{\ \mu\kappa\nu} = W^{\kappa}_{\ \mu\nu\kappa} = 0 \, ,
\]
but other identities fulfilled by the Weyl tensor cannot even be formulated here, because there is \textbf{no} metric tensor to lower the upper index!  Furthermore, the Ricci tensor splits into its symmetric -- $K_{\mu\nu}$ -- and its skew-symmetric -- $F_{\mu\nu}$ -- parts, according to $R_{\mu\nu} = K_{\mu\nu} + F_{\mu\nu}$. Finally, decomposition of the Riemann tensor into three irreducible parts (i.e.:~$K$, $F$ and $W$) enables us to rewrite the gravitational component \eqref{varaffa} of the generating formula in a way similar to Lemma \ref{lemma5.1} in the purely metric theory (cf.~\cite{Univer} and see Appendix \ref{krzywizna}):
\begin{theorem}
\label{th 6.2}
The following identity holds:
\begin{eqnarray}
 \partial_{\nu}\left(\mathcal{P}_{\kappa}^{\ \lambda\mu\nu}\, \delta\Gamma^{\kappa}_{\ \lambda\mu}   \right) = \left(\nabla_{\nu}\mathcal{P}_{\kappa}^{\ \lambda\mu\nu} \right)\, \delta\Gamma^{\kappa}_{\ \lambda\mu} + \pi^{\mu\nu}\, \delta K_{\mu\nu}  +\chi^{\mu\nu}\, \delta F_{\mu\nu} -2\,  \Omega_{\kappa}^{\ \lambda[\mu\nu]}\, \delta W^{\kappa}_{\ \lambda\mu\nu}\, , \qquad
\end{eqnarray}
where
\begin{eqnarray*}
\pi^{\mu \nu} &=&-\frac 23 \mathcal{P}_{\kappa}^{\  \kappa (\mu \nu)} \, ,  \\
\chi^{\mu \nu} &=&- \frac 25  \mathcal{P}_{\kappa}^{\ \kappa [\mu \nu]} \, ,
\end{eqnarray*}
and $\Omega$ is the remaining, traceless part of $\mathcal{P}$.
\end{theorem}
Applying above decomposition to the variational formula \eqref{varaffa}, we obtain:
\begin{eqnarray}
 \delta\Lag_A(K_{\mu\nu}, F_{\mu\nu}, W^{\kappa}_{\ \lambda\mu\nu} , \phi , \nabla \phi)  &=& \partial_{\nu}\left(p^{\nu}\, \delta\phi \right)  +  \left(\nabla_{\nu}\mathcal{P}_{\kappa}^{\ \lambda\mu\nu} \right)\, \delta\Gamma^{\kappa}_{\ \lambda\mu} + \pi^{\mu\nu}\, \delta K_{\mu\nu}+ \nonumber \\
  &&  +\chi^{\mu\nu}\, \delta F_{\mu\nu} -2\,  \Omega_{\kappa}^{\ \lambda[\mu\nu]}\, \delta W^{\kappa}_{\ \lambda\mu\nu}  \, . \label{varaff2}
\end{eqnarray}
Because $\Gamma$ enters here through “covariant derivatives” of matter field, the gravitational part of the Euler-Lagrange  equation reads:
\begin{equation}\label{ELaff}
    \nabla_{\nu}\mathcal{P}_{\kappa}^{\ \lambda\mu\nu} = \frac{\partial\Lag_A}{\partial \Gamma^{\kappa}_{\ \lambda\mu}} = \frac{\partial\Lag_A}{\partial(\nabla_{\alpha} \phi)}
    \frac{\partial(\nabla_{\alpha} \phi)}{\partial \Gamma^{\kappa}_{\ \lambda\mu}}=`` p\cdot \phi " \, .
\end{equation}

Mathematical structure and the physical interpretation of a general theory \eqref{varaff2} will be discussed in another paper. Here, we limit ourselves to the conventional General Relativity Theory. This means, as proved in Section \ref{univ-aff}, that we consider Lagrangians which depend upon curvature {\em via} the symmetric part of the Ricci tensor, exclusively: $\Lag_A=\Lag_A(K_{\mu\nu} , \phi , \nabla \phi)$. Consequently, we have:
\begin{eqnarray}
\chi^{\mu\nu}&=&\frac{\partial\Lag_A}{\partial F_{\mu\nu}} =0\, , \\
-2\, \Omega_{\kappa}^{\ \lambda[\mu\nu]}&=&\frac{\partial\Lag_A}{\partial W^{\kappa}_{\ \lambda\mu\nu}} =0 \, , \\
 \pi^{\mu\nu} &=& \frac{\partial\Lag_A}{\partial K_{\mu\nu}} \, , \\
{\cal P}_\kappa^{\ \lambda\mu\nu} = \pi_{\kappa}^{\ \lambda\mu\nu} &=&
{\pi}^{\lambda\mu} \, \delta_{\kappa}^{\nu} -\delta_{\kappa}^{(\lambda}\, \pi^{\mu)\nu} \, .
\end{eqnarray}
This means, that we obtain formula \eqref{var Lag P 2-}, used already in Section \ref{univ-P}:
\begin{eqnarray}
\label{varaff3}
\delta\Lag_A  =\partial_{\nu}\left(p^{\nu}\, \delta\phi \right) +  \left(\nabla_{\nu}\pi_{\kappa}^{\ \lambda\mu\nu} \right)\, \delta\Gamma^{\kappa}_{\ \lambda\mu} + \pi^{\mu\nu}\, \delta K_{\mu\nu}\, , \qquad
\end{eqnarray}
or equivalently:
\begin{eqnarray}
\label{varaff4}
\delta\Lag_A  =\partial_{\nu}\left(\pi_{\kappa}^{\ \lambda\mu\nu}\, \delta \Gamma^{\kappa}_{\ \lambda\mu}+ p^{\nu}\, \delta\phi \right)\, ,
\end{eqnarray}
cf. formula \eqref{var Lag P 1}. The gravitational field equations are:
\begin{eqnarray}
 \nabla_{\nu}\pi_{\kappa}^{\ \lambda\mu\nu}  &=&\frac{\partial \Lag_A}{\partial\Gamma^{\kappa}_{\ \lambda\mu}}=
 \frac{\partial\Lag_A}{\partial \nabla_{\alpha}\phi}\, \frac{\partial \nabla_{\alpha}\phi}{\partial \Gamma^{\kappa}_{\ \lambda\mu}}= p^{\alpha}\, \frac{\partial \nabla_{\alpha}\phi}{\partial \Gamma^{\kappa}_{\ \lambda\mu}} = `` p\cdot \phi "\, ,  \label{eqa1} \\
 \pi^{\mu\nu} &=& \frac{\partial \Lag_A}{\partial K_{\mu\nu}} \, . \label{eqa2}
\end{eqnarray}
The first equation enables us to find the non-metricity tensor $N$ as a difference between our affine connection $\Gamma$ and the Levi-Civita connection $\mGamma$ -- see formula~\eqref{Gamma-gen}. Then, equation \eqref{varaff4} reads:
\begin{eqnarray}
\label{varaff5}
\delta\Lag_A  =\partial_{\nu}\left(\pi_{\kappa}^{\ \lambda\mu\nu}\, \delta \mGamma^{\kappa}_{\ \lambda\mu}+ \pi_{\kappa}^{\ \lambda\mu\nu}\, \delta N^{\kappa}_{\ \lambda\mu} + p^{\nu}\, \delta\phi \right)\, . \qquad
\end{eqnarray}
The first term can be calculated from identity \eqref{varG} which was already presented in Section \ref{sec var}:
\begin{eqnarray}
 \partial_{\nu}\left(\pi_{\kappa}^{\ \lambda\mu\nu}\, \delta \mGamma^{\kappa}_{\ \lambda\mu}  \right) &=&  \delta {\cal L}_H + \frac 1{16 \pi}\, \kolo{\cal G}^{\mu\nu}\,  \delta g_{\mu\nu} \, .
\end{eqnarray}
For the second term we use the “inverted” Lemma \ref{lemma non-metricity} and obtain:
\begin{eqnarray}
 \partial_{\nu}\left( \pi_{\kappa}^{\ \lambda\mu\nu}\, \delta N^{\kappa}_{\ \lambda\mu}  \right) =  \partial_{\kappa} \left( {\cal R}^{\mu\nu\kappa}\, \delta g_{\mu\nu} \right) - \delta \left[ \mnabla_\kappa  \mathcal{R}_{\sigma}^{\ \sigma\kappa} \right]\, , \qquad
\end{eqnarray}
where \begin{eqnarray}
\mathcal{R}^{\mu\nu\kappa} &=& \frac{\sqrt{|\det g|}}{16\pi}\, \left[N^{\kappa\mu\nu} - g^{\kappa(\mu}\, N^{\ \nu)\sigma}_{ \sigma}  +\frac 12 \left( N_{\sigma}^{\ \sigma\kappa}  - N^{\kappa\sigma}_{\ \ \sigma}\right)\, g^{\mu\nu} \right] \, .
\end{eqnarray}
Analogously, the “inverted”  Lemma \ref{lemma varGamma} gives us:
\begin{eqnarray*}
     \partial_{\kappa} \left( {\cal R}^{\mu\nu\kappa}\, \delta g_{\mu\nu} \right) &=&\mathcal{P}_{\ \ \lambda}^{ \mu\nu}\,  \delta \mGamma^{\lambda}_{\  \mu\nu} + \left(\mnabla_{\kappa} \mathcal{R}^{\mu\nu\kappa}  \right)\,  \delta g_{\mu\nu}      \, ,
\end{eqnarray*}
where:
\begin{eqnarray}
 \mathcal{P}^{\mu\nu}_{\ \ \lambda} := \mathcal{R}^{\  \mu\nu}_{ \lambda} +  \mathcal{R}^{\  \nu\mu}_{ \lambda}\, .
\end{eqnarray}
This way, variational formula \eqref{varaff5} can be rewritten as follows:
\begin{eqnarray}
\delta\Lag_A  &=&\partial_{\nu}\left( p^{\nu}\, \delta\phi \right) +  \delta {\cal L}_H  +  \mathcal{P}_{\ \ \lambda}^{ \mu\nu}\,  \delta \mGamma^{\lambda}_{\  \mu\nu} +
\nonumber \\
&&+\left(\frac 1{16 \pi}\, \kolo{\cal G}^{\mu\nu} + \mnabla_{\kappa} \mathcal{R}^{\mu\nu\kappa}  \right)\,  \delta g_{\mu\nu} - \delta \left[ \mnabla_\kappa  \mathcal{R}_{\sigma}^{\ \sigma\kappa} \right] \, . \label{varaff6}
\end{eqnarray}
Putting the complete variation $\delta \left[\mnabla R \right]$ on the left-hand side, we finally obtain the metric Lagrangian (cf.~\eqref{Pala-def}):
\[
\Lag_g:= \Lag_A + \mnabla_\kappa  \mathcal{R}_{\sigma}^{\ \sigma\kappa}
\]
together with its variation:
\begin{eqnarray*}
\delta \Lag_g &=& \partial_{\nu}\left( p^{\nu}\, \delta\phi \right) +  \delta {\cal L}_H  +  \mathcal{P}_{\ \ \lambda}^{ \mu\nu}\,  \delta \mGamma^{\lambda}_{\  \mu\nu} +\left(\frac 1{16 \pi}\, \kolo{\cal G}^{\mu\nu} + \mnabla_{\kappa} \mathcal{R}^{\mu\nu\kappa}  \right)\,  \delta g_{\mu\nu} \, .
\end{eqnarray*}
Let us define the matter Lagrangian density $\Lag_{matt} := \Lag_g - \Lag_H$. It satisfies equation:
\begin{eqnarray}
 \delta\left(\Lag_g - \Lag_H \right) &=& \delta \Lag_{matt} =  \partial_{\nu}\left( p^{\nu}\, \delta\phi \right)  +  \mathcal{P}_{\ \ \lambda}^{ \mu\nu}\,  \delta \mGamma^{\lambda}_{\  \mu\nu} +  \nonumber \\
 &&+\left(\frac 1{16 \pi}\, \kolo{\cal G}^{\mu\nu} + \mnabla_{\kappa} \mathcal{R}^{\mu\nu\kappa}  \right)\,  \delta g_{\mu\nu}\, .
\end{eqnarray}
This equation implies the gravitational field equations \eqref{aaa} -- \eqref{Leg-m-a2} in the purely metric formulation.

We see that the current standard version of general relativity uses only one of the three sectors of complete affine theory: the one corresponding to the symmetric part $K_{\mu\nu}$ of the Ricci. The remaining sectors of the curvature, namely $F$ and $W$, are “fallow” and not used. It is shown in \cite{uni} that they can describe properly the electromagnetic field (cf.~also the paper by H.~Weyl \cite{Weyl}) and the dark matter, respectively.

\section{Conclusions}\label{concl}

The “Palatini method of variation” leads to the standard, metric formulation of General Relativity Theory only for a special class of matter fields: where the matter Lagrangian density $\Lag_{matt} = \Lag_{matt}(\phi , \partial \phi , g)$ does not depend upon connection coefficients (i.e. does not depend upon derivatives of the metric tensor). But, for a generic matter field we have ${\cal L}_{matt} = {\cal L}_{matt} (\phi , \mnabla \phi , g) = {\cal L}_{matt} (\phi , \partial \phi , g, \partial g)$ and, whence, the “Palatini method” leads {\em a priori} to the non-metric connection $\Gamma \ne \ \ \mGamma$. In this paper we propose to accept this non-metricity of the connection. We show that the theory, rewritten in terms of this connection -- although entirely equivalent to the conventional GR theory -- acquires a much simpler mathematical structure. Physically, this means that similarly as “mass generates curvature” (i.e. departure from flatness), the matter sensitivity towards the connection generates the departure of the connection from metricity. The two spacetime geometric structures, namely: metric and connection, acquire an independent status. The relation between them is not assumed {\em a priori} but is governed by field equations. The current theory of gravitational interactions uses only one sector of a generic mathematical affine theory discussed here. It is likely that the remaining sectors can be used to describe other physical (possibly non-gravitational) interactions.

\appendix

\section{Proof of Lemma \ref{lemma varGamma}}
\label{proof lemma varGamma}

Using explicit formula for the metric connection coefficients, we obtain:
\begin{eqnarray*}
    \mathcal{P}_{\ \ \lambda}^{ \mu\nu}\,  \delta \stackrel{\circ \ }{\Gamma^{\lambda}}_{ \mu\nu} &=& \mathcal{P}_{\ \ \lambda}^{ \mu\nu}\, \delta \left[\frac 12\, g^{\lambda\kappa}\, \left(g_{\kappa\mu , \nu} +  g_{\kappa\nu , \mu} - g_{\mu\nu , \kappa}   \right)   \right]  = \mathcal{P}_{\ \ \lambda}^{ \mu\nu}\, \delta \left[ g^{\lambda\kappa}\, \left(g_{\kappa\mu , \nu}  -  \frac 12\, g_{\mu\nu , \kappa}   \right)   \right]= \\
    &=&  \mathcal{P}_{\ \ \lambda}^{ \mu\nu}\,  \left(g_{\kappa\mu , \nu}  -   \frac 12\, g_{\mu\nu , \kappa}   \right)\,  \delta g^{\lambda\kappa}  +\mathcal{P}^{\mu\nu\kappa } \, \delta \left(g_{\kappa\mu , \nu}  -  \frac 12\, g_{\mu\nu , \kappa}   \right)= \\
    &=& - \mathcal{P}_{\ \ \lambda}^{ \mu\nu}\, \left(g_{\kappa\mu , \nu}  -  \frac 12\, g_{\mu\nu , \kappa}   \right)\,  g^{\lambda\alpha}\,  g^{\kappa\beta}\,  \delta g_{\alpha\beta}  +\mathcal{P}^{\mu\nu\kappa} \, \delta \left(g_{\kappa\mu , \nu}  -  \frac 12\, g_{\mu\nu , \kappa}   \right)= \\
    &=& -  \mathcal{P}^{\mu\nu\alpha}\, \stackrel{\circ\ }{\Gamma^{\beta}}_{ \mu\nu}\,  \delta g_{\alpha\beta}    + \partial_{\kappa} \left[ \left(  \mathcal{P}^{\nu \kappa \mu} - \frac 12\,  \mathcal{P}^{\mu\nu\kappa}\right)\, \delta g_{\mu\nu} \right]  + \\
    && -   \left[\partial_{\kappa}\left(  \mathcal{P}^{\nu \kappa \mu} - \frac 12\,  \mathcal{P}^{\mu\nu\kappa}\right)\right] \, \delta g_{\mu\nu}= \\
    &=& - \left({\cal R}^{\mu\nu\kappa}_{\ \ \  ,\kappa} + \mathcal{P}^{\alpha \beta\mu}\, \stackrel{\circ\ }{\Gamma^{\nu}}_{ \alpha\beta} \right)\,  \delta g_{\mu\nu}   +\partial_{\kappa}\left( {\cal R}^{\mu\nu\kappa}\, \delta g_{\mu\nu} \right)  \, .
\end{eqnarray*}
Taking into account that both $\mathcal{R}$ and $\cal P$ are tensor densities, we have:
\begin{eqnarray*}
    \mnabla_{\kappa} \mathcal{R}^{\mu \nu \kappa} &=&\mathcal{R}^{\mu \nu \kappa}_{\ \ \ ,\kappa} - \stackrel{\circ \ }{\Gamma^{\sigma}}_{ \sigma \kappa}\, \mathcal{R}^{\mu \nu \kappa} +  \stackrel{\circ \ }{\Gamma^{\mu}}_{ \sigma \kappa}\, \mathcal{R}^{\sigma \nu \kappa} +\stackrel{\circ \ }{\Gamma^{\nu}}_{ \sigma \kappa}\, \mathcal{R}^{\mu \sigma \kappa} + \stackrel{\circ \ }{\Gamma^{\kappa}}_{ \sigma \kappa}\, \mathcal{R}^{\mu \nu \sigma} = \\
    &=&\mathcal{R}^{\mu \nu \kappa}_{\ \ \ ,\kappa} + \frac 12\, \stackrel{\circ \ }{\Gamma^{\mu}}_{ \sigma \kappa}\, \mathcal{P}^{\sigma\kappa\nu } + \frac 12\, \stackrel{\circ \ }{\Gamma^{\nu}}_{ \sigma \kappa}\, \mathcal{P}^{\sigma\kappa\mu} = \mathcal{R}^{\mu \nu \kappa}_{\ \ \ ,\kappa} +  \mathcal{P}^{ \sigma\kappa(\mu}\stackrel{\circ \ }{\Gamma^{\nu)}}_{ \sigma \kappa}    \, ,
\end{eqnarray*}
which implies the thesis of the Lemma.

\section{Proof of Lemma \ref{lemma non-metricity} }
\label{proof lemma non-metricity}

Using formula \eqref{R delta pi} we obtain:
\begin{eqnarray*}
    {\cal R}^{\mu\nu\kappa}\, \delta g_{\mu\nu}&=&\frac{16 \pi}{\sqrt{|\det g|}}\,  \left(\frac 12 \, \mathcal{R}_{\lambda}^{\  \lambda\kappa}\,  g_{\alpha\beta} - \mathcal{R}_{\alpha\beta}^{\ \ \  \kappa} \right)\, \delta \pi^{\alpha\beta} =\\
    &=& {16 \pi}\,  \left(\frac 12 \, R_{\lambda}^{\  \lambda\kappa}\,  g_{\alpha\beta} - R_{\alpha\beta}^{\ \ \  \kappa} \right)\, \delta \pi^{\alpha\beta} =\\
    &=& \delta \left\{
    {16 \pi}\,  \left(\frac 12 \, R_{\lambda}^{\  \lambda\kappa}\,  g_{\alpha\beta} - R_{\alpha\beta}^{\ \ \  \kappa} \right)\,  \pi^{\alpha\beta}\right\}+\\
    &&  -     \pi^{\alpha\beta} \delta \left[
    {16 \pi}\,  \left(\frac 12 \, R_{\lambda}^{\  \lambda\kappa}\,  g_{\alpha\beta} - R_{\alpha\beta}^{\ \ \  \kappa} \right)
    \right]= \\
    &=&      \delta  {\cal R}_{\sigma}^{\ \sigma \kappa} - \pi^{\alpha\beta}\, \delta \left[ 16 \pi\, \left( \frac 12\, R_{\sigma}^{\ \sigma \kappa}\,  g_{\alpha\beta} - R_{\alpha\beta}^{\ \ \ \kappa} \right) \right] \, ,
\end{eqnarray*}
and, consequently:
\begin{eqnarray*}
     \partial_{\kappa}\left( {\cal R}^{\mu\nu\kappa}\delta g_{\mu\nu} \right)   =  \delta \left[ {\cal R}_{\sigma \ \ ,\kappa}^{\ \sigma \kappa}\right] +\partial_{\kappa}\left\{\pi^{\alpha\beta}\, \delta \left[16 \pi\, \left( R_{\alpha\beta}^{\ \ \ \kappa} - \frac 12\, R_{\sigma}^{\ \sigma \kappa}\,  g_{\alpha\beta}  \right) \right]\right\} \, .
\end{eqnarray*}
Since $\mathcal{R}^{\mu\nu\kappa}\,  g_{\mu\nu}$ is a vector density, we have:
\[
    {\cal R}_{\sigma \ \ ,\kappa}^{\ \sigma \kappa}= \
    \mnabla_\kappa  \mathcal{R}_{\sigma}^{\ \sigma\kappa} \, ,
\]
and, whence:
\begin{eqnarray*}
   \partial_{\kappa}\left( {\cal R}^{\mu\nu\kappa}\delta g_{\mu\nu} \right)  =  \delta \left[ \mnabla_\kappa  \mathcal{R}_{\sigma}^{\ \sigma\kappa} \right] +  \partial_{\kappa} \left\{\pi^{\alpha\beta}\, \delta \left[16 \pi\, \left( R_{\alpha\beta}^{\ \ \ \kappa} - \frac 12\, R_{\sigma}^{\ \sigma \kappa}\,  g_{\alpha\beta}  \right) \right]\right\} \, .
\end{eqnarray*}
Now, it remains to check that the quantity $N$ defined by \eqref{def: tensor N obraz metryczny} satisfies identity
\begin{eqnarray*}
    \pi_{\kappa}^{\ \lambda \mu \nu}\, \delta N^{\kappa}_{\ \lambda \mu} :=   \pi^{\alpha\beta}\, \delta \left[ 16 \pi \left(R_{\alpha\beta}^{\ \ \ \nu} - \frac 12\, R_{\sigma}^{\ \sigma \nu}\,  g_{\alpha\beta}  \right) \right]\, ,
\end{eqnarray*}
which ends the proof.

\section{Different representations of the curvature and its decomposition into irreducible components}
\label{krzywizna}

The Riemann tensor of a symmetric connection, defined by \eqref{def riemann}, splits into the three irreducible components: the symmetric ($K_{\mu\nu}$) and the skew-symmetric ($F_{\mu\nu}$) parts of the Ricci tensor $R_{\mu\nu}:=R^{\kappa}_{\ \mu \kappa \nu}= K_{\mu\nu} + F_{\mu\nu}$, and the remaining traceless part $W^{\kappa}_{\ \lambda\mu\nu}$, according to the following decomposition formula:
\begin{eqnarray}
R^{\kappa}_{\ \lambda \mu \nu} &=& \frac 13 \left(\delta^{\kappa}_{\mu}\, K_{\lambda \nu} - \delta^{\kappa}_{\nu}\, K_{\lambda \mu} \right) +\frac 15 \left(2\,\delta^{\kappa}_{\lambda}\, F_{\mu \nu} + \delta^{\kappa}_{\mu}\, F_{\lambda \nu} - \delta^{\kappa}_{\nu}\, F_{\lambda \mu} \right) + W^{\kappa}_{\ \lambda \mu \nu} \, .\qquad
\label{rozklad: tensor riemanna}
\end{eqnarray}
The Riemann tensor satisfies the following two algebraic identities:
\begin{equation}\label{idenR}
    R^{\kappa}_{\ \lambda \mu \nu} = - R^{\kappa}_{\ \lambda \nu \mu} \, ,
    \quad   R^{\kappa}_{\ [\lambda \mu \nu]}=0 \, .
\end{equation}
(The square bracket denotes skew-symmetrization.)

As will be seen in the next Appendix  (see also \cite{Senger1}, \cite{Univer}), for many calculable purposes it is useful to represent the same geometric object (the curvature) by a different, but totally equivalent, {\em curvature tensor},
defined in terms of the first jet of the connection by the following formula:
\begin{eqnarray}
K^{\kappa}_{\ \lambda \mu \nu} :=  \Gamma^{\kappa}_{\ \lambda \mu \nu} - \Gamma^{\kappa}_{\ (\lambda \mu \nu)} + \Gamma^{\sigma}_{\ \lambda \mu}\, \Gamma^{\kappa}_{\ \nu \sigma} -  \Gamma^{\sigma}_{\ (\lambda \mu}\, \Gamma^{\kappa}_{\ \nu) \sigma} \, .  \qquad
\label{krzywizna kijowskiego 1}
\end{eqnarray}
The equivalence between the two objects is assured by the following identities:
\begin{eqnarray}
    K^{\kappa}_{\ \lambda \mu \nu} &=& -\frac 23 \, R^{\kappa}_{\ (\lambda \mu) \nu}\, , \quad  R^{\kappa}_{\ \lambda \mu \nu} = -2 \, K^{\kappa}_{\ \lambda [\mu \nu]}\, .
\end{eqnarray}
(The round bracket denotes symmetrization.)

The curvature tensor $K$ satisfies the following two algebraic identities:
\begin{equation}\label{idenK}
        K^{\kappa}_{\ \lambda \mu \nu} = K^{\kappa}_{\  \mu \lambda \nu} \, ,  \qquad   K^{\kappa}_{\ (\lambda \mu \nu)}=0 \, ,
\end{equation}
analogous to \eqref{idenR}, and can be decomposed into its irreducible components in a way analogous  to \eqref{rozklad: tensor riemanna}, namely:
\begin{eqnarray}
    K^{\kappa}_{\ \lambda \mu \nu} = -\frac 19\left(\delta^{\kappa}_{\lambda}\, K_{\mu \nu} + \delta^{\kappa}_{\mu}\, K_{\lambda \nu} - 2 \delta^{\kappa}_{\nu}\, K_{\lambda \mu} \right) - \frac 15\left(\delta^{\kappa}_{\lambda}\, F_{\mu \nu} + \delta^{\kappa}_{\mu}\, F_{\lambda \nu} \right) + U^{\kappa}_{\  \lambda \mu  \nu} \, ,
    \label{krzywizna kijowskiego 2}
\end{eqnarray}
where
\begin{eqnarray*}
    U^{\kappa}_{\  \lambda \mu  \nu} := -\frac 23\, W^{\kappa}_{\  (\lambda \mu)  \nu} \, .
\end{eqnarray*}

\section{Proof of Theorem \ref{th 6.2} }
\label{proof th 6.2}

Because curvature tensor is equivalent with Riemann, we can assume that the invariant Lagrangian density $\Lag_A (\Gamma,\, \partial\Gamma,\, \phi, \partial \phi)$ depends upon derivatives of the connection coefficients {\em via} the tensor $K^{\kappa}_{\ \lambda \mu \nu}$, i.e. {\em via} the expression “$\Gamma^{\kappa}_{\ \lambda \mu \nu} - \Gamma^{\kappa}_{\ (\lambda \mu \nu)}$” which is: 1) symmetric in indices $(\lambda, \mu)$ and 2) its totally symmetric part $\Gamma^{\kappa}_{\ (\lambda \mu \nu)}~-~\Gamma^{\kappa}_{\ (\lambda \mu \nu)}=0$ vanishes identically. Being equal to the derivative of $\Lag_A$ with respect to this quantity, the momentum $\mathcal{P}_{\kappa}^{\ \lambda \mu \nu}$ must fulfil {\em a priori} identities similar to \eqref{idenK}, namely:
\begin{equation}\label{idenP}
    \mathcal{P}_{\kappa}^{\ \lambda \mu \nu} = \mathcal{P}_{\kappa}^{\  \mu \lambda \nu} \, , \quad  \mathcal{P}_{\kappa}^{\ (\lambda \mu \nu)}=0 \, .
\end{equation}
We have, therefore:
\begin{eqnarray*}
    \partial_{\nu}\left(\mathcal{P}_{\kappa}^{\ \lambda \mu \nu}\, \delta \Gamma^{\kappa}_{\ \lambda \mu}\right)  &=& \mathcal{P}_{\kappa \ \ \ ,\nu}^{\ \lambda \mu \nu}\, \delta \Gamma^{\kappa}_{\ \lambda \mu} + \mathcal{P}_{\kappa}^{\ \lambda \mu \nu}\, \delta \Gamma^{\kappa}_{\ \lambda \mu,\nu}= \\
    &=& \mathcal{P}_{\kappa \ \ \ ,\nu}^{\ \lambda \mu \nu}\, \delta \Gamma^{\kappa}_{\ \lambda \mu} + \mathcal{P}_{\kappa}^{\ \lambda \mu \nu}\, \delta \left(\Gamma^{\kappa}_{\ \lambda \mu,\nu} - \Gamma^{\kappa}_{\ (\lambda \mu,\nu)} \right)     \, ,
\end{eqnarray*}
because the last term vanishes. Using definition~(\ref{krzywizna kijowskiego 1}) of the curvature tensor, we get:
\begin{eqnarray*}
    \partial_{\nu}\left(\mathcal{P}_{\kappa}^{\ \lambda \mu \nu}\, \delta \Gamma^{\kappa}_{\ \lambda \mu}\right)&=& \mathcal{P}_{\kappa}^{\ \lambda \mu \nu}\, \delta K^{\kappa}_{\ \lambda \mu,\nu} + \mathcal{P}_{\kappa \ \ \ ,\nu}^{\ \lambda \mu \nu}\, \delta \Gamma^{\kappa}_{\ \lambda \mu} - \mathcal{P}_{\kappa}^{\ \lambda \mu \nu}\, \delta \left(  \Gamma^{\sigma}_{\ \lambda \mu}\, \Gamma^{\kappa}_{\ \nu \sigma}  -  \Gamma^{\sigma}_{\ (\lambda \mu}\, \Gamma^{\kappa}_{\ \nu) \sigma} \right) =\\
    &=& \mathcal{P}_{\kappa}^{\ \lambda \mu \nu}\, \delta K^{\kappa}_{\ \lambda \mu,\nu} + \mathcal{P}_{\kappa \ \ \ ,\nu}^{\ \lambda \mu \nu}\, \delta \Gamma^{\kappa}_{\ \lambda \mu}   - \mathcal{P}_{\kappa}^{\ \lambda \mu \nu}\, \delta  \left(  \Gamma^{\sigma}_{\ \lambda \mu}\, \Gamma^{\kappa}_{\ \nu \sigma} \right)  \, .
\end{eqnarray*}
But, we have:
\begin{lemma}
    For a tensor density $\mathcal{P}_{\kappa}^{\ \lambda \mu \nu}$, the following identity holds:
\begin{eqnarray*}
    \mathcal{P}_{\kappa \ \ \ ,\nu}^{\ \lambda \mu \nu}\, \delta \Gamma^{\kappa}_{\ \lambda \mu} - \mathcal{P}_{\kappa}^{\ \lambda \mu \nu}\, \delta \left(  \Gamma^{\sigma}_{\ \lambda \mu}\, \Gamma^{\kappa}_{\ \nu \sigma} \right) = \nabla_{\nu}\, \mathcal{P}_{\kappa}^{\ \lambda \mu \nu}\, \delta \Gamma^{\kappa}_{\ \lambda \mu}\, .
\end{eqnarray*}
\end{lemma}
The proof is entirely computational:
\begin{eqnarray*}
    \mathcal{P}_{\kappa \ \ \ ,\nu}^{\ \lambda \mu \nu}\, \delta \Gamma^{\kappa}_{\ \lambda \mu} - \mathcal{P}_{\kappa}^{\ \lambda \mu \nu}\, \delta \left(  \Gamma^{\sigma}_{\ \lambda \mu}\, \Gamma^{\kappa}_{\ \nu \sigma}   \right) &=&   \mathcal{P}_{\kappa \ \ \ ,\nu}^{\ \lambda \mu \nu}\, \delta \Gamma^{\kappa}_{\ \lambda \mu} - \mathcal{P}_{\kappa}^{\ \lambda \mu \nu}\,      \Gamma^{\sigma}_{\ \lambda \mu}\, \delta  \Gamma^{\kappa}_{\ \nu \sigma}   - \mathcal{P}_{\kappa}^{\ \lambda \mu \nu}\,\Gamma^{\kappa}_{\ \nu \sigma}     \, \delta   \Gamma^{\sigma}_{\ \lambda \mu}   = \\
    &=& \left(\mathcal{P}_{\kappa \ \ \ ,\nu}^{\ \lambda \mu \nu} - \mathcal{P}_{\sigma}^{\ \lambda \mu\nu}\, \Gamma^{\sigma}_{\ \kappa \nu} - \mathcal{P}_{\kappa}^{\ \sigma\nu \lambda}\, \Gamma^{\mu}_{ \ \sigma \nu} \right)\,\delta \Gamma^{\kappa}_{\ \lambda \mu} \, .
\end{eqnarray*}
Using again \eqref{idenP}, the last term can be rewritten as follows:
\begin{eqnarray*}
 - \mathcal{P}_{\kappa}^{\ \sigma\nu \lambda}\, \Gamma^{\mu}_{ \ \sigma \nu} \,\delta \Gamma^{\kappa}_{\ \lambda \mu} &=&  \mathcal{P}_{\kappa}^{\ \sigma\nu \lambda}\, \Gamma^{\mu}_{ \ \nu\lambda} \,\delta \Gamma^{\kappa}_{\ \sigma \mu}
 + \mathcal{P}_{\kappa}^{\ \sigma\nu \lambda}\, \Gamma^{\mu}_{ \ \lambda\sigma } \,\delta \Gamma^{\kappa}_{\ \nu\mu}=
 \\
   &=&
 \mathcal{P}_{\kappa}^{\ \lambda\nu \sigma}\, \Gamma^{\mu}_{ \ \nu\sigma} \,\delta \Gamma^{\kappa}_{\ \lambda \mu}
 + \mathcal{P}_{\kappa}^{\ \sigma\lambda\nu }\, \Gamma^{\mu}_{ \ \nu\sigma } \,\delta \Gamma^{\kappa}_{\ \lambda\mu}= \\
 &=& \mathcal{P}_{\kappa}^{\ \lambda\sigma\nu }\, \Gamma^{\mu}_{ \ \sigma\nu} \,\delta \Gamma^{\kappa}_{\ \lambda \mu}
 + \mathcal{P}_{\kappa}^{\ \nu\lambda\sigma }\, \Gamma^{\mu}_{ \ \sigma\nu } \,\delta \Gamma^{\kappa}_{\ \lambda\mu}
 \,  ,
\end{eqnarray*}
and, whence:
\begin{eqnarray*}
  \mathcal{P}_{\kappa \ \ \ ,\nu}^{\ \lambda \mu \nu}\, \delta \Gamma^{\kappa}_{\ \lambda \mu} - \mathcal{P}_{\kappa}^{\ \lambda \mu \nu}\, \delta \left(  \Gamma^{\sigma}_{\ \lambda \mu}\, \Gamma^{\kappa}_{\ \nu \sigma}   \right)  =  \left(\mathcal{P}_{\kappa \ \ \ ,\nu}^{\ \lambda \mu \nu} - \mathcal{P}_{\sigma}^{\ \lambda \mu\nu}\, \Gamma^{\sigma}_{\ \kappa \nu}+  \mathcal{P}_{\kappa}^{\   \lambda \sigma\nu}\, \Gamma^{\mu}_{ \ \sigma \nu} + \mathcal{P}_{\kappa}^{\ \nu \lambda \sigma}\, \Gamma^{\mu}_{ \ \sigma \nu} \right)\,\delta \Gamma^{\kappa}_{\ \lambda \mu} \, .
\end{eqnarray*}
On the other hand, we have:
\begin{eqnarray*}
    \left(\nabla_{\nu}\, \mathcal{P}_{\kappa}^{\ \lambda \mu \nu}\right)\, \delta \Gamma^{\kappa}_{\ \lambda \mu} &=& \left(\mathcal{P}_{\kappa \ \ \ ,\nu}^{\ \lambda \mu \nu} - \mathcal{P}_{\sigma}^{\ \lambda \mu\nu}\, \Gamma^{\sigma}_{\ \kappa \nu}  + \mathcal{P}_{\kappa}^{\   \sigma \mu \nu}\, \Gamma^{\mu}_{ \ \sigma \nu} +\mathcal{P}_{\kappa}^{\  \lambda \sigma \nu}\, \Gamma^{\mu}_{ \ \sigma \nu} \right)\,\delta \Gamma^{\kappa}_{\ \lambda \mu}\, ,
\end{eqnarray*}
which ends the proof of the Lemma.

\

Inserting formula:
\begin{eqnarray}
        \partial_{\nu}\left(\mathcal{P}_{\kappa}^{\ \lambda \mu \nu}\, \delta \Gamma^{\kappa}_{\ \lambda \mu}\right)  = \mathcal{P}_{\kappa}^{\ \lambda \mu \nu}\, \delta K^{\kappa}_{\ \lambda \mu \nu} + \nabla_{\nu}\,\mathcal{P}_{\kappa}^{\ \lambda \mu \nu}\, \delta \Gamma^{\kappa}_{\ \lambda \mu}  \, , \nonumber \\
                \label{var: lagranzjan afiniczny z krzywizna kijowskiego}
\end{eqnarray}
into \eqref{krzywizna kijowskiego 2} we obtain:
\begin{eqnarray*}
     \partial_{\nu}\left(\mathcal{P}_{\kappa}^{\ \lambda \mu \nu}\, \delta \Gamma^{\kappa}_{\ \lambda \mu}\right)  =\pi^{\mu\nu}\, \delta K_{\mu \nu}  + \Omega_{\kappa}^{\ \lambda \mu\nu}\, \delta U^{\kappa}_{\ \lambda \mu\nu} +     \chi^{\mu \nu}\, \delta F_{\mu \nu}+\left(\nabla_{\nu}\, \mathcal{P}_{\kappa}^{\ \lambda \mu \nu}\right)\, \delta \Gamma^{\kappa}_{\ \lambda \mu}\, ,
\end{eqnarray*}
where we have defined:
\begin{eqnarray}
\pi^{\mu \nu} &=&-\frac 23 \mathcal{P}_{\kappa}^{\  \kappa (\mu \nu)} \, ,
\label{def: ped pi jako slad pedu calP} \\
\chi^{\mu \nu} &=&- \frac 25  \mathcal{P}_{\kappa}^{\ \kappa [\mu \nu]} \,
,\label{def: ped chi jako slad pedu calP}
\end{eqnarray}
whereas $\Omega$ represents the remaining, traceless part of ${\cal P}$. This finally gives rise to the following decomposition of the momentum ${\cal P}$, analogous to \eqref{idenR} and \eqref{idenK}:
\begin{eqnarray}
    \mathcal{P}_{\kappa}^{\ \lambda \mu \nu}= -\frac 12
\left(\delta^{\lambda}_{\kappa}\, \pi^{\mu \nu }+\delta^{\mu}_{\kappa}\,
\pi^{\lambda \nu} - 2\delta^{\nu}_{\kappa}\, \pi^{  \lambda\mu} \right)  -\frac
12 \left(\delta^{\lambda}_{\kappa}\, \chi^{\mu \nu }+\delta^{\mu}_{\kappa}\,
\chi^{\lambda \nu}  \right) + \Omega_{\kappa}^{\ \lambda \mu \nu}\, .
\label{eq: rozklad pedu P}
\end{eqnarray}

The techniques used above prove that the use of the curvature tensor $K^{\kappa}_{\ \lambda \mu \nu}$, instead of the Riemann tensor $R^{\kappa}_{\ \lambda \mu \nu}$ simplifies considerably description of the canonical structure of the theory. Indeed, the canonical momentum \eqref{DefP} is directly equal to the derivative of the Lagrangian density $\Lag$ with respect to the curvature:
\begin{eqnarray}
\label{DefP1}
     \mathcal{P}_{\kappa}^{\ \lambda\mu\nu} := \frac {\partial \Lag_A}{\partial \Gamma^{\kappa}_{\ \lambda\mu,\nu}}  = \frac {\partial \Lag_A}{\partial K^{\kappa}_{\ \lambda\mu\nu}}
     \, ,
\end{eqnarray}
whereas derivative with respect to the Riemann tensor would need further symmetrization, which obscures considerably this, relatively simple and transparent, structure.

\section{Relation between the metric and non-metric Ricci tensors}
\label{proof R}
Using definition \eqref{def: tensor Ricciego} of the Ricci tensor $R_{\mu\nu}$ of the connection $\Gamma^{\kappa}_{\ \mu \nu} = \ \mGamma^{\kappa}_{\ \mu \nu} + N^{\kappa}_{\ \mu \nu}$, we can easily decompose the complete Ricci into the metric Ricci\  $\kolo R$ and the remaining part, which depends upon non-metricity and its derivatives:
\begin{eqnarray*}
    R_{\mu \nu}&=&  -\Gamma^{\kappa}_{\ \kappa\mu, \nu} + \Gamma^{\kappa}_{\ \mu \nu, \kappa} - \Gamma^{\sigma}_{\  \kappa\mu}\, \Gamma^{\kappa}_{\ \nu \sigma} + \Gamma^{\sigma}_{\ \mu \nu}\, \Gamma^{\kappa}_{\ \kappa \sigma} = \nonumber  \\
    &=&\kolo R_{\mu\nu} -  \underline{N^{\kappa}_{\ \kappa\mu, \nu}} + \underbrace{N^{\kappa}_{\ \mu \nu, \kappa}} -\, \underbrace{\mGamma^{\sigma}_{\  \kappa\mu}\, N^{\kappa}_{\ \nu \sigma}} +  \underline{\mGamma^{\sigma}_{\ \mu \nu}\, N^{\kappa}_{\ \kappa \sigma}} - \underbrace{N^{\sigma}_{\  \kappa\mu}\, \mGamma^{\kappa}_{\ \nu \sigma}} +\underbrace{ N^{\sigma}_{\ \mu \nu}\, \mGamma^{\kappa}_{\ \kappa \sigma}} + \\
    &&- N^{\sigma}_{\  \kappa\mu}\, N^{\kappa}_{\ \nu \sigma} + N^{\sigma}_{\ \mu \nu}\, N^{\kappa}_{\ \kappa \sigma} \, .
\end{eqnarray*}
The terms, which have been marked, gather to the covariant derivatives of the non-metricity $N$:
\begin{eqnarray}
    R_{\mu \nu}&=&\kolo R_{\mu\nu} -  \underline{\mnabla_{\nu} N^{\kappa}_{\ \kappa\mu }} + \underbrace{\mnabla_{\kappa} N^{\kappa}_{\ \mu \nu }}  - N^{\sigma}_{\  \kappa\mu}\, N^{\kappa}_{\ \nu \sigma} + N^{\sigma}_{\ \mu \nu}\, N^{\kappa}_{\ \kappa \sigma} \, .
    \label{rozklad ricci}
\end{eqnarray}

\section{Non-metricity tensor $N$}
\label{calP N}
To find the non-metricity tensor $N$ we have to solve equation \eqref{N=lincalP}:
\begin{eqnarray*}
    {\cal P}_{\ \ \lambda}^{\mu\nu} &=& \pi^{\mu\alpha} N^{\nu}_{ \ \lambda\alpha} +
    \pi^{\nu\alpha} N^{\mu}_{ \ \lambda\alpha} - \pi^{\mu\nu} N^{\alpha}_{ \ \lambda\alpha} - \frac 12
    \left( \delta^\mu_\lambda\, N^{\nu}_{ \ \alpha\beta} + \delta^\nu_\lambda \, N^{\mu}_{ \ \alpha\beta}
    \right) \pi^{\alpha\beta} \, .
\end{eqnarray*}
First, we calculate the trace putting $\nu=\lambda $:
\begin{eqnarray*}
    \mathcal{P}^{\mu\lambda}_{\ \ \lambda} = -\frac 32\, \frac{\sqrt{|\det g|}}{  16\pi} \,N^{\mu \alpha}_{ \ \ \alpha }  \, ,
\end{eqnarray*}
and contracting the equation \eqref{N=lincalP} with metric tensor $g_{\mu\nu}$:
\begin{equation*}
      {\cal P}_{\ \sigma \lambda}^{\sigma} = \frac{\sqrt{|\det g|}}{16\pi}\left[ -2 N^{\alpha}_{ \ \lambda\alpha}   -   N^{ \ \ \alpha}_{ \lambda \alpha}   \right]\, .
\end{equation*}
Then, a simple algebra leads to the final result:
\begin{eqnarray*}
 N^{\kappa}_{\ \lambda\mu} = \frac{8\pi}{\sqrt{|\det g|}}\, g^{\kappa \sigma} \, \left[ \mathcal{P}_{\sigma \lambda \mu} + \mathcal{P}_{\sigma \mu \lambda} - \mathcal{P}_{ \lambda \mu \sigma}
 + g_{\lambda \mu}\, \left(\mathcal{P}_{\sigma \nu}^{\ \ \nu} - \frac 12\, \mathcal{P}^{\nu}_{\ \nu \sigma}  \right)  +\mathcal{P}^{\nu}_{\ \nu (\lambda}\, g_{\mu)\sigma} - \frac 23\,  g_{\sigma(\lambda}\, \mathcal{P}_{\mu) \nu}^{\ \ \ \nu}
\right] \, .
\end{eqnarray*}

\section{Example of transition from metric to Palatini picture }
\label{example}
Take the following matter Lagrangian density for the vector field $\phi^{\mu}$:
\begin{eqnarray*}
 \Lag_{matt}(g,\, \mnabla \phi) = \frac{\sqrt{|\det g|} }{16 \pi} \,\left( \mnabla_{\alpha}\phi^{\beta}\right)\, \left( \mnabla^{\alpha}\phi_{\beta}\right)\, .
\end{eqnarray*}
To find the corresponding Palatini Lagrangian density \eqref{LP-1}, we have to calculate the non-metricity tensor $N$ as a function of $g, \phi, \nabla\phi$ from equation \eqref{N=lincalP}. The idea of solving is quite easy (see Appendix \ref{calP N}), but, unfortunately, calculations are difficult and cumbersome. Using the symbolic calculation package\footnote{We used the  suite of free packages \textit{xAct} version 1.2.0 (especially\textit{ xTensor} package) prepared by Jos\'e M. Mart\'in-Garc\'ia (see www.xAct.es) for \textit{Mathematica} version 12.2 environment.}, we obtain the final result. In terms of the following auxiliary quantities:
 \begin{eqnarray*}
   \phi^2 &=&\phi^{\alpha}\, \phi_{\alpha}\, , \\
    D_1 &=& 3+2\, \phi^2\, , \\
     D_2&=& -\phi^4-\phi^2+3\, , \\
   B_{\alpha}^{\ \beta} &=& \nabla_{\alpha}\phi^{\beta}\, , \\
    E_{\mu}&=& B_{\mu}^{\ \alpha}\,\phi_{\alpha}\, , \\
    F^{\mu}&=&\phi^{\alpha}\, B_{\alpha}^{\ \mu}\, , \\
    A&=& \phi^{\alpha}\, B_{\alpha}^{\ \beta}\, \phi_{\beta}=E_{\mu}\,\phi^{\mu}=\phi_{\mu}\, F^{\mu}\, , \\
    B &=& B_{\alpha}^{\ \alpha}\,  ,
   \end{eqnarray*}
the Palatini Lagrangian is equals:
\begin{eqnarray*}
  \Lag_P(g,\, \nabla \phi) = \frac{\sqrt{|\det g|} }{32\pi\, D_1\, D_2\, (\phi^2-1)} && \bigg[ D_1\,  D_2\,\left(\phi^2-2\right) \, B_{\kappa\lambda} \, B^{\kappa\lambda}    + D_1\,  D_2\,   \phi^2  \, B_{\kappa\lambda}\, B^{\lambda\kappa} + \\
 &&+4 A^2 \left(\phi^2-1\right) \phi^2 +4 D_2\,  \left(\phi^2-1\right)\, A\, B +  \\
 && -  \phi^2\,  D_2\,  \left(\phi^2-1\right)\, B^2  - D_1\, \phi^4\,  F_{\mu}\, F^{\mu} +  \\
 &&- 2 D_1\,      \left(\phi^2-2\right) \, \phi^2\,   E_{\mu}\, F^{\mu}  -D_1\,  \left(\phi^2-2\right)^2 \, E_{\mu}\, E^{\mu}  \bigg] \, .
\end{eqnarray*}

\section{Transition from Palatini to metric picture -- example}
\label{example2}
Take as the Palatini Lagrangian for the vector field $X^{\mu}$ the quadratic expression:
\begin{eqnarray*}
 \Lag_P(g,\, \nabla X)=\frac{\sqrt{|\det g|} }{16\pi} \,\left( \nabla_{\alpha} X^{\mu}\right)\, \left( \nabla_{\beta}X^{\nu}\right)\, g^{\alpha\beta}\, g_{\mu\nu} \, .
\end{eqnarray*}
In a way analogous to the one used in Appendix \ref{example}, using the symbolic computations provided by Mathematica, we obtain the final result in terms of the auxiliary quantities:
\begin{eqnarray*}
   X^2 &=&X^{\alpha}\,X_{\alpha}\, , \\
  \widetilde{D}_1&=&2X^2-3\, , \\
   \widetilde{D}_2&=&X^4-X^2-3\, , \\
   \widetilde{B}_{\mu}^{\ \alpha} &=& \mnabla_{\mu}X^{\alpha}\, , \\
   \widetilde{E}_{\mu} &=&  \widetilde{B}_{\mu}^{\ \alpha}\, X_{\alpha}\, , \\
   \widetilde{F}^{\alpha}&=&  X^{\mu}\, \widetilde{B}_{\mu}^{\ \alpha}\, , \\   \widetilde{A}&=&X_{\alpha}\,\widetilde{F}^{\alpha}=\widetilde{E}_{\alpha}\,X^{\alpha} \, , \\
   \widetilde{B}&=&\widetilde{B}_{\alpha}^{\ \alpha} \, .
   \end{eqnarray*}
The corresponding metric matter Lagrangian defined by equation \eqref{lagp3} has the following form:
     \begin{eqnarray*}
 \Lag_{matt}(g,\, \mnabla X) = \frac{\sqrt{|\det g|} }{32\pi\, \widetilde{D}_1\, \widetilde{D}_2\, (X^2+1)}  && \bigg[ \widetilde{D}_1 \,  \widetilde{D}_2 \,\left(X^2+2\right) \, \widetilde{B}_{\kappa\lambda}\, \widetilde{B}^{\kappa\lambda}  +  \widetilde{D}_1 \,  \widetilde{D}_2 \,X^2\,  \widetilde{B}_{\lambda\kappa}  \, \widetilde{B}^{\kappa\lambda} +   \\
 &&+ 4 \left(X^2+1\right)\, X^2\, \widetilde{A}^2 + 4 \widetilde{D}_2\,  \left(X^2+1\right)\,  \widetilde{A}\, \widetilde{B}  +   \\
 &&-  X^2\,    \widetilde{D}_2 \,  \left(X^2+1\right)\,  \widetilde{B}^2  - \widetilde{D}_1\, \left(X^2-2\right)^2\,    \widetilde{E}_{\mu}\, \widetilde{E}^{\mu}  +   \\
 &&- 2  \widetilde{D}_1 \, \left(X^2+2\right)\,  X^2\,  \widetilde{E}_{\mu}\, \widetilde{F}^{\mu}  - \widetilde{D}_1   \,X^4\, \widetilde{F}_{\mu}\, \widetilde{F}^{\mu}\bigg]
\, .
\end{eqnarray*}

\end{document}